\documentclass[final]{IEEEtran}
\usepackage{array}
\newcolumntype{P}[1]{>{\centering\arraybackslash}p{#1}}
\newcolumntype{M}[1]{>{\centering\arraybackslash}m{#1}}
\usepackage{amsthm,amssymb,amsmath,graphicx,multirow,color,amsfonts}%
\usepackage[update,prepend]{epstopdf}
\usepackage[latin1]{inputenc}
\usepackage{tikz}
\usepackage{bbm} % for \mathbbm{1}
\usepackage{pdfpages}
\usepackage{subfig}
\usepackage{comment}
\usepackage{stfloats}
\usepackage{multicol}

\usepackage{cite}
\usepackage[justification=centering]{caption}
\usepackage{textcomp}
\usepackage{psfrag}
\usepackage{arydshln}
\usepackage{url}
\usepackage{soul}
\usepackage{graphicx,color}
\usepackage[nolist]{acronym}
\usepackage{algorithm,algorithmic} %algorithm
\usepackage{mathtools,lipsum}
\usepackage{cuted}
\usepackage{amsmath}
\newtheorem{proposition}{Proposition}  
\usepackage{mathrsfs}
\usepackage[capitalise]{cleveref}
\Crefname{equation}{Eq.\!}{Eqs.\!}
\Crefname{figure}{Fig.\!}{Figs.\!}
\Crefname{tabular}{Tab.\!}{Tabs.\!}
\Crefname{section}{Section\!}{Sections.\!}

\def\nb0{{\mathbf{0}}}
\def\nb1{{\mathbf{1}}}

\newtheorem{lemma}{Lemma}

\newtheorem{definition}{Definition}

\newtheorem{theorem}{Theorem}

\begin{document}
%\pagenumbering{gobble}
\graphicspath{{./Figures/}}
	\begin{acronym}

\acro{5G-NR}{5G New Radio}
\acro{3GPP}{3rd Generation Partnership Project}
\acro{ABS}{aerial base station}
\acro{AC}{address coding}
\acro{ACF}{autocorrelation function}
\acro{ACR}{autocorrelation receiver}
\acro{ADC}{analog-to-digital converter}
\acrodef{aic}[AIC]{Analog-to-Information Converter}     
\acro{AIC}[AIC]{Akaike information criterion}
\acro{aric}[ARIC]{asymmetric restricted isometry constant}
\acro{arip}[ARIP]{asymmetric restricted isometry property}

\acro{ARQ}{Automatic Repeat Request}
\acro{AUB}{asymptotic union bound}
\acrodef{awgn}[AWGN]{Additive White Gaussian Noise}     
\acro{AWGN}{additive white Gaussian noise}

\acro{APSK}[PSK]{asymmetric PSK} 

\acro{waric}[AWRICs]{asymmetric weak restricted isometry constants}
\acro{warip}[AWRIP]{asymmetric weak restricted isometry property}
\acro{BCH}{Bose, Chaudhuri, and Hocquenghem}        
\acro{BCHC}[BCHSC]{BCH based source coding}
\acro{BEP}{bit error probability}
\acro{BFC}{block fading channel}
\acro{BG}[BG]{Bernoulli-Gaussian}
\acro{BGG}{Bernoulli-Generalized Gaussian}
\acro{BPAM}{binary pulse amplitude modulation}
\acro{BPDN}{Basis Pursuit Denoising}
\acro{BPPM}{binary pulse position modulation}
\acro{BPSK}{Binary Phase Shift Keying}
\acro{BPZF}{bandpass zonal filter}
\acro{BSC}{binary symmetric channels}              
\acro{BU}[BU]{Bernoulli-uniform}
\acro{BER}{bit error rate}
\acro{BS}{base station}
\acro{BW}{BandWidth}
\acro{BLLL}{ binary log-linear learning }

\acro{CP}{Cyclic Prefix}
\acrodef{cdf}[CDF]{cumulative distribution function}   
\acro{CDF}{Cumulative Distribution Function}
\acrodef{c.d.f.}[CDF]{cumulative distribution function}
\acro{CCDF}{complementary cumulative distribution function}
\acrodef{ccdf}[CCDF]{complementary CDF}               
\acrodef{c.c.d.f.}[CCDF]{complementary cumulative distribution function}
\acro{CD}{cooperative diversity}

\acro{CDMA}{Code Division Multiple Access}
\acro{ch.f.}{characteristic function}
\acro{CIR}{channel impulse response}
\acro{cosamp}[CoSaMP]{compressive sampling matching pursuit}
\acro{CR}{cognitive radio}
\acro{cs}[CS]{compressed sensing}                   
\acrodef{cscapital}[CS]{Compressed sensing} %will not include it in the list
\acrodef{CS}[CS]{compressed sensing}
\acro{CSI}{channel state information}
\acro{CCSDS}{consultative committee for space data systems}
\acro{CC}{convolutional coding}
\acro{Covid19}[COVID-19]{Coronavirus disease}

\acro{DAA}{detect and avoid}
\acro{DAB}{digital audio broadcasting}
\acro{DCT}{discrete cosine transform}
\acro{dft}[DFT]{discrete Fourier transform}
\acro{DR}{distortion-rate}
\acro{DS}{direct sequence}
\acro{DS-SS}{direct-sequence spread-spectrum}
\acro{DTR}{differential transmitted-reference}
\acro{DVB-H}{digital video broadcasting\,--\,handheld}
\acro{DVB-T}{digital video broadcasting\,--\,terrestrial}
\acro{DL}{DownLink}
\acro{DSSS}{Direct Sequence Spread Spectrum}
\acro{DFT-s-OFDM}{Discrete Fourier Transform-spread-Orthogonal Frequency Division Multiplexing}
\acro{DAS}{Distributed Antenna System}
\acro{DNA}{DeoxyriboNucleic Acid}

\acro{EC}{European Commission}
\acro{EED}[EED]{exact eigenvalues distribution}
\acro{EIRP}{Equivalent Isotropically Radiated Power}
\acro{ELP}{equivalent low-pass}
\acro{eMBB}{Enhanced Mobile Broadband}
\acro{EMF}{ElectroMagnetic Field}
\acro{EU}{European union}
\acro{EI}{Exposure Index}
\acro{eICIC}{enhanced Inter-Cell Interference Coordination}

\acro{FC}[FC]{fusion center}
\acro{FCC}{Federal Communications Commission}
\acro{FEC}{forward error correction}
\acro{FFT}{fast Fourier transform}
\acro{FH}{frequency-hopping}
\acro{FH-SS}{frequency-hopping spread-spectrum}
\acrodef{FS}{Frame synchronization}
\acro{FSsmall}[FS]{frame synchronization}  
\acro{FDMA}{Frequency Division Multiple Access}

\acro{GA}{Gaussian approximation}
\acro{GF}{Galois field }
\acro{GG}{Generalized-Gaussian}
\acro{GIC}[GIC]{generalized information criterion}
\acro{GLRT}{generalized likelihood ratio test}
\acro{GPS}{Global Positioning System}
\acro{GMSK}{Gaussian Minimum Shift Keying}
\acro{GSMA}{Global System for Mobile communications Association}
\acro{GS}{ground station}
\acro{GMG}{ Grid-connected MicroGeneration}

\acro{HAP}{high altitude platform}
\acro{HetNet}{Heterogeneous network}

\acro{IDR}{information distortion-rate}
\acro{IFFT}{inverse fast Fourier transform}
\acro{iht}[IHT]{iterative hard thresholding}
\acro{i.i.d.}{independent, identically distributed}
\acro{IoT}{Internet of Things}                      
\acro{IR}{impulse radio}
\acro{lric}[LRIC]{lower restricted isometry constant}
\acro{lrict}[LRICt]{lower restricted isometry constant threshold}
\acro{ISI}{intersymbol interference}
\acro{ITU}{International Telecommunication Union}
\acro{ICNIRP}{International Commission on Non-Ionizing Radiation Protection}
\acro{IEEE}{Institute of Electrical and Electronics Engineers}
\acro{ICES}{IEEE international committee on electromagnetic safety}
\acro{IEC}{International Electrotechnical Commission}
\acro{IARC}{International Agency on Research on Cancer}
\acro{IS-95}{Interim Standard 95}

\acro{KPI}{Key Performance Indicator}

\acro{LEO}{low earth orbit}
\acro{LF}{likelihood function}
\acro{LLF}{log-likelihood function}
\acro{LLR}{log-likelihood ratio}
\acro{LLRT}{log-likelihood ratio test}
\acro{LoS}{Line-of-Sight}
\acro{LRT}{likelihood ratio test}
\acro{wlric}[LWRIC]{lower weak restricted isometry constant}
\acro{wlrict}[LWRICt]{LWRIC threshold}
\acro{LPWAN}{Low Power Wide Area Network}
\acro{LoRaWAN}{Low power long Range Wide Area Network}
\acro{NLoS}{Non-Line-of-Sight}
\acro{LiFi}[Li-Fi]{light-fidelity}
 \acro{LED}{light emitting diode}
 \acro{LABS}{LoS transmission with each ABS}
 \acro{NLABS}{NLoS transmission with each ABS}

\acro{MB}{multiband}
\acro{MC}{macro cell}
\acro{MDS}{mixed distributed source}
\acro{MF}{matched filter}
\acro{m.g.f.}{moment generating function}
\acro{MI}{mutual information}
\acro{MIMO}{Multiple-Input Multiple-Output}
\acro{MISO}{multiple-input single-output}
\acrodef{maxs}[MJSO]{maximum joint support cardinality}                       
\acro{ML}[ML]{maximum likelihood}
\acro{MMSE}{minimum mean-square error}
\acro{MMV}{multiple measurement vectors}
\acrodef{MOS}{model order selection}
\acro{M-PSK}[${M}$-PSK]{$M$-ary phase shift keying}                       
\acro{M-APSK}[${M}$-PSK]{$M$-ary asymmetric PSK} 
\acro{MP}{ multi-period}
\acro{MINLP}{mixed integer non-linear programming}

\acro{M-QAM}[$M$-QAM]{$M$-ary quadrature amplitude modulation}
\acro{MRC}{maximal ratio combiner}                  
\acro{maxs}[MSO]{maximum sparsity order}                                      
\acro{M2M}{Machine-to-Machine}                                                
\acro{MUI}{multi-user interference}
\acro{mMTC}{massive Machine Type Communications}      
\acro{mm-Wave}{millimeter-wave}
\acro{MP}{mobile phone}
\acro{MPE}{maximum permissible exposure}
\acro{MAC}{media access control}
\acro{NB}{narrowband}
\acro{NBI}{narrowband interference}
\acro{NLA}{nonlinear sparse approximation}
\acro{NLOS}{Non-Line of Sight}
\acro{NTIA}{National Telecommunications and Information Administration}
\acro{NTP}{National Toxicology Program}
\acro{NHS}{National Health Service}

\acro{LOS}{Line of Sight}

\acro{OC}{optimum combining}                             
\acro{OC}{optimum combining}
\acro{ODE}{operational distortion-energy}
\acro{ODR}{operational distortion-rate}
\acro{OFDM}{Orthogonal Frequency-Division Multiplexing}
\acro{omp}[OMP]{orthogonal matching pursuit}
\acro{OSMP}[OSMP]{orthogonal subspace matching pursuit}
\acro{OQAM}{offset quadrature amplitude modulation}
\acro{OQPSK}{offset QPSK}
\acro{OFDMA}{Orthogonal Frequency-division Multiple Access}
\acro{OPEX}{Operating Expenditures}
\acro{OQPSK/PM}{OQPSK with phase modulation}

\acro{PAM}{pulse amplitude modulation}
\acro{PAR}{peak-to-average ratio}
\acrodef{pdf}[PDF]{probability density function}                      
\acro{PDF}{probability density function}
\acrodef{p.d.f.}[PDF]{probability distribution function}
\acro{PDP}{power dispersion profile}
\acro{PMF}{probability mass function}                             
\acrodef{p.m.f.}[PMF]{probability mass function}
\acro{PN}{pseudo-noise}
\acro{PPM}{pulse position modulation}
\acro{PRake}{Partial Rake}
\acro{PSD}{power spectral density}
\acro{PSEP}{pairwise synchronization error probability}
\acro{PSK}{phase shift keying}
\acro{PD}{power density}
\acro{8-PSK}[$8$-PSK]{$8$-phase shift keying}
\acro{PPP}{Poisson point process}
\acro{PCP}{Poisson cluster process}
 
\acro{FSK}{Frequency Shift Keying}

\acro{QAM}{Quadrature Amplitude Modulation}
\acro{QPSK}{Quadrature Phase Shift Keying}
\acro{OQPSK/PM}{OQPSK with phase modulator }

\acro{RD}[RD]{raw data}
\acro{RDL}{"random data limit"}
\acro{ric}[RIC]{restricted isometry constant}
\acro{rict}[RICt]{restricted isometry constant threshold}
\acro{rip}[RIP]{restricted isometry property}
\acro{ROC}{receiver operating characteristic}
\acro{rq}[RQ]{Raleigh quotient}
\acro{RS}[RS]{Reed-Solomon}
\acro{RSC}[RSSC]{RS based source coding}
\acro{r.v.}{random variable}                               
\acro{R.V.}{random vector}
\acro{RMS}{root mean square}
\acro{RFR}{radiofrequency radiation}
\acro{RIS}{Reconfigurable Intelligent Surface}
\acro{RNA}{RiboNucleic Acid}
\acro{RRM}{Radio Resource Management}
\acro{RUE}{reference user equipments}
\acro{RAT}{radio access technology}
\acro{RB}{resource block}

\acro{SA}[SA-Music]{subspace-augmented MUSIC with OSMP}
\acro{SC}{small cell}
\acro{SCBSES}[SCBSES]{Source Compression Based Syndrome Encoding Scheme}
\acro{SCM}{sample covariance matrix}
\acro{SEP}{symbol error probability}
\acro{SG}[SG]{sparse-land Gaussian model}
\acro{SIMO}{single-input multiple-output}
\acro{SINR}{signal-to-interference plus noise ratio}
\acro{SIR}{signal-to-interference ratio}
\acro{SISO}{Single-Input Single-Output}
\acro{SMV}{single measurement vector}
\acro{SNR}[\textrm{SNR}]{signal-to-noise ratio} 
\acro{sp}[SP]{subspace pursuit}
\acro{SS}{spread spectrum}
\acro{SW}{sync word}
\acro{SAR}{specific absorption rate}
\acro{SSB}{synchronization signal block}
\acro{SR}{shrink and realign}

\acro{tUAV}{tethered Unmanned Aerial Vehicle}
\acro{TBS}{terrestrial base station}

\acro{uUAV}{untethered Unmanned Aerial Vehicle}
\acro{PDF}{probability density functions}

\acro{PL}{path-loss}

\acro{TH}{time-hopping}
\acro{ToA}{time-of-arrival}
\acro{TR}{transmitted-reference}
\acro{TW}{Tracy-Widom}
\acro{TWDT}{TW Distribution Tail}
\acro{TCM}{trellis coded modulation}
\acro{TDD}{Time-Division Duplexing}
\acro{TDMA}{Time Division Multiple Access}
\acro{Tx}{average transmit}

\acro{UAV}{Unmanned Aerial Vehicle}
\acro{uric}[URIC]{upper restricted isometry constant}
\acro{urict}[URICt]{upper restricted isometry constant threshold}
\acro{UWB}{ultrawide band}
\acro{UWBcap}[UWB]{Ultrawide band}   
\acro{URLLC}{Ultra Reliable Low Latency Communications}
         
\acro{wuric}[UWRIC]{upper weak restricted isometry constant}
\acro{wurict}[UWRICt]{UWRIC threshold}                
\acro{UE}{User Equipment}
\acro{UL}{UpLink}

\acro{WiM}[WiM]{weigh-in-motion}
\acro{WLAN}{wireless local area network}
\acro{wm}[WM]{Wishart matrix}                               
\acroplural{wm}[WM]{Wishart matrices}
\acro{WMAN}{wireless metropolitan area network}
\acro{WPAN}{wireless personal area network}
\acro{wric}[WRIC]{weak restricted isometry constant}
\acro{wrict}[WRICt]{weak restricted isometry constant thresholds}
\acro{wrip}[WRIP]{weak restricted isometry property}
\acro{WSN}{wireless sensor network}                        
\acro{WSS}{Wide-Sense Stationary}
\acro{WHO}{World Health Organization}
\acro{Wi-Fi}{Wireless Fidelity}

\acro{sss}[SpaSoSEnc]{sparse source syndrome encoding}

\acro{VLC}{Visible Light Communication}
\acro{VPN}{Virtual Private Network} 
\acro{RF}{Radio Frequency}
\acro{FSO}{Free Space Optics}
\acro{IoST}{Internet of Space Things}

\acro{GSM}{Global System for Mobile Communications}
\acro{2G}{Second-generation cellular network}
\acro{3G}{Third-generation cellular network}
\acro{4G}{Fourth-generation cellular network}
\acro{5G}{Fifth-generation cellular network}	
\acro{gNB}{next-generation Node-B Base Station}
\acro{NR}{New Radio}
\acro{UMTS}{Universal Mobile Telecommunications Service}
\acro{LTE}{Long Term Evolution}

\acro{QoS}{Quality of Service}
\end{acronym}
	
	%% EMF definitions
\newcommand{\SAR} {\mathrm{SAR}}
\newcommand{\WBSAR} {\mathrm{SAR}_{\mathsf{WB}}}
\newcommand{\gSAR} {\mathrm{SAR}_{10\si{\gram}}}
\newcommand{\Sab} {S_{\mathsf{ab}}}
\newcommand{\Eavg} {E_{\mathsf{avg}}}
\newcommand{\ft}{f_{\textsf{th}}}
\newcommand{\alphatf}{\alpha_{24}}

%\title{
%Maximum Throughput Multi-Hop Routing for THz and RF Transmission
%}

\title{THz/RF Multi-Hop Routing Throughput: Performance, Optimization, and Application}

\author{
Zhengying Lou, Baha Eddine Youcef Belmekki,~\IEEEmembership{Senior Member,~IEEE,} and Mohamed-Slim Alouini, {\em Fellow, IEEE}

\thanks{Zhengying Lou and Mohamed-Slim Alouini are with the King Abdullah
University of Science and Technology, Thuwal 23955, Saudi Arabia
(e-mail: {zhengying.lou, slim.alouini}@kaust.edu.sa). Baha Eddine
Youcef Belmekki is with the School of Engineering and Physical Sciences,
Heriot-Watt University, Edinburgh EH14 4AS, United Kingdom (e-mail:
b.belmekki@hw.ac.uk).}
\vspace{-8mm}
}

\maketitle

\begin{abstract}
\color{black} Terahertz (THz) communication offers a promising solution for high-throughput wireless systems. However, the severe path loss of THz signals raises concerns about its effectiveness compared to radio frequency (RF) communication. 
In this article, we establish the first stochastic geometry (SG)-based analytical framework for routing in THz systems.
We develop a stepwise optimization approach to maximize throughput, including power allocation, relay selection, and number of hops design. 
Analytical expressions for throughput and coverage probability are derived under the SG framework, enabling low complexity and scalable performance evaluation.
Numerical results show that the proposed stepwise-optimal routing strategies not only outperform existing SG-based methods but also approach the ideal upper bound. Moreover, we compare the throughput and coverage performance of THz and RF routing and demonstrate the applications of the proposed analytical framework and routing strategies in system parameter design and unmanned aerial vehicle networks.
\end{abstract}


\section{Introduction}
Terahertz (THz) frequency band communication is playing a pivotal role in wireless communication, addressing the ever-increasing demand for wireless data transmission \cite{faisal2020ultramassive,belmekki2023harnessing,elzanaty2021towards}. 
The THz band offers a large bandwidth and a high-frequency reuse factor, making it highly suitable for achieving ultra-high data transmission rates. In this context, the THz spectrum emerges as an ideal choice to enable the seamless exchange of real-time data, including high-definition video, sensor information, and map data \cite{sarieddeen2020next,wan2021terahertz}. Furthermore, the THz networks hold significant potential for applications requiring low-latency transmission of large data packets, such as enhancing vehicle communication systems with autonomous driving technology \cite{zaid2023aerial}, improving user experiences in augmented and virtual reality \cite{chaccour2020ruin}, and expanding the capabilities of non-terrestrial networks \cite{yuan2022secure}.

\par
{\color{black} When compared to traditional radio frequency (RF) transmission, THz communication offers an ultra-wide bandwidth, which, combined with beamforming techniques, enables high data rates, precise directional transmission, and enhanced spectral efficiency \cite{elayan2019terahertz,rappaport2013millimeter}.}
However, the propagation distance of THz signals is significantly limited by water molecule absorption \cite{chattopadhyay2015compact}. In regions with ample moisture, the effective transmission distance of RF suppresses that of THz networks. Furthermore, THz signals exhibit significant penetration and path loss, resulting in reduced signal strength compared to RF signals. 
Therefore, comparing the transmission performance of THz and RF networks is an intriguing topic.

\par
To facilitate a fair comparison, we consider the following scenario. The source node chooses either RF or THz links to transmit signals to the target node, utilizing intermediate relay nodes to circumvent obstacles and improve signal reliability \cite{akyildiz2022terahertz}. The total transmission power of the aforementioned multi-hop routing and the distance between the source and target nodes should be consistent. Throughput, which is defined as the maximum achievable rate of data transmission over a fading channel \cite{al2020throughput}, is determined by both the bandwidth and path loss. Thus, throughput is suitable for evaluating the routing performance. 
In summary, we compare the throughput performance between routing with RF transmission and routing with THz transmission under fixed transmission power and distance.

\subsection{Related Works}
{\color{black} 
Numerous studies have considered the analysis and comparison of the transmission performance of THz and RF routing. Among them, simulation-based routing and fixed-topology routing are the two most commonly adopted methods \cite{han2022terahertz,serghiou2022terahertz,boulogeorgos2020outage,farrag2021outage,bhardwaj2021performance}. 
However, the methods and frameworks also struggle to address the challenges of comparison fairness. 
{\color{black} Moreover, the conclusions drawn from these studies are often limited and not widely applicable.}

Simulation-based routing methods require a large number of testing rounds to address the randomness of the channel, with the average performance being calculated as the final result\cite{han2022terahertz}. Moreover, in each round of testing, simulations need to re-determine routing nodes and estimate link performance\cite{serghiou2022terahertz}. Therefore, evaluation through simulation is typically computationally expensive. Furthermore, results obtained through numerical simulations may lack scalability as they become inapplicable when system parameters, such as communication distance, are changed. 

Fixed topology routing methods can provide closed-form analytical solutions for low computational complexity performance evaluation \cite{boulogeorgos2020outage,farrag2021outage}. Authors in \cite{boulogeorgos2020outage} derived closed-form expressions for the outage probability of dual-hop THz systems, while \cite{farrag2021outage} proposed an optimum power allocation scheme. Considering both THz and RF relay nodes, the authors in \cite{bhardwaj2021performance} derived analytical expressions for the SNR and ergodic capacity. However, fixed topology routing is also not widely applicable because once the routing path changes, the previously obtained evaluations no longer hold. Therefore, we need a mathematical tool that can establish an analytical framework for stochastic topology, meeting the requirements of wide applicability.

Stochastic geometry (SG) is one of the most suitable mathematical tools for the analysis of large-scale networks with stochastic topology \cite{wang2021stochastic}. Specifically, relay nodes are modeled as spatial distributions to enable analyzability.
In the field of SG, there are several studies involving the analysis and comparison of network performance under different transmission frequency bands \cite{sayehvand2020interference,lou2023coverage,belbase2018coverage}. 
In \cite{sayehvand2020interference}, authors focused on analyzing the downlink coverage probability of a hybrid THz and RF network. Then, a dual-hop decode-and-forward routing protocol in a hybrid THz and RF network was proposed in \cite{lou2023coverage}. On the other hand, authors in \cite{belbase2018coverage} compared the dual-hop coverage probability of millimeter-wave and RF networks.

However, these studies restrict the number of hops to single-hop or double-hop, thereby failing to meet the requirement for comparison fairness. For instance, in medium or long-distance communication, THz communication faces more significant attenuation compared to RF due to the large distance. However, this doesn't necessarily prove that RF has an advantage over the THz frequency band in medium or long-distance communication scenarios. THz networks can mitigate the impact of single-hop attenuation by employing more hops in the routing. Therefore, it is more reasonable to compare the two transmission modes under the routing with an optimal number of hops. Based on the analysis above, it is urgent to provide a strategy that simultaneously meets the requirements of wide applicability and comparison fairness.}

\subsection{Contribution}
{\color{black}
In this paper, to provide a scalable, efficient, and fair analytical framework for performance evaluation, SG-based routing for both THz and RF networks is investigated. To the best of our knowledge, this is the first work that applies SG to study routing strategies in THz networks. The contributions of this paper are summarized as follows.}

\begin{itemize}
{\color{black} \item 
To promote comparison fairness, we design maximum throughput routing schemes separately for THz and RF networks. These schemes consist of a stepwise optimization framework, comprising power allocation, relay selection, and number of hops design strategies. Furthermore, we provide an unreachable upper bound for the routing throughput in an ideal scenario, and the throughput corresponding to the proposed stepwise-optimal routing strategy can approach this upper bound. 

\item To meet the need for wide applicability, we introduce the SG framework to derive analytical results for routing throughput and coverage probability. A stepwise-suboptimal routing strategy is further proposed to enhance tractability, whose throughput performance is proven to be a tight lower bound for that of the stepwise-routing strategy.
Among existing analyzable routing strategies, stepwise-optimal and stepwise-suboptimal strategies exhibit significantly higher throughput than long-hop and short-hop strategies commonly studied in the literature.

\item Through numerical results, we compare the throughput and coverage performance of THz and RF transmission. We also demonstrate how the analytical results of coverage probability can be applied to transmission power design. Additionally, the application of this routing strategy in unmanned aerial vehicle (UAV) networks is showcased through simple extensions.}
\end{itemize}

{\color{black}
The rest of this paper is organized as follows. The system model is presented in Section \ref{section2}. In Section \ref{section3}, a stepwise optimization approach is proposed to maximize throughput. The analytical expressions for throughput and coverage probability are derived in Section \ref{section4}. 
Section \ref{section5} shows the comprehensive numerical results. Section \ref{section6} discusses open issues and future research directions. Section \ref{section7} concludes this paper.}

\section{System Model}\label{section2}

\subsection{Network Model}
We consider a scenario where a source node communicates with a target node with the help of THz and RF devices, which serve as relay nodes to allow multi-hop routing, as shown in Fig.\ref{fig:sys}.
A transmission between two nodes is called a hop, and a routing with $K_{Q}$ hops ($K_{Q} \geq 1$) contains $K_{Q}-1$ $Q$ relay nodes, where $Q = \{  \rm{THz}, \rm{RF} \}$. Without loss of generality, we set the position of the source node at the origin and the direction from the source node to the target node as the positive $x$-axis, and we set the target node at a distance $R$ from the source (i.e., the origin). 
Given that all nodes are located on a two-dimensional plane, the locations of THz and RF relay nodes form two independent homogeneous Poisson point processes (PPPs) denoted by $\Phi_{\rm{THz}}$ and $\Phi_{\rm{RF}}$ with density $\lambda_{\rm{THz}}$ and $\lambda_{\rm{RF}}$, respectively. 
{\color{black} The network nodes are modeled as quasi-static during each transmission interval. This modeling approach reflects practical scenarios such as UAV hovering and vehicular relaying along roads, where the relative positions of nodes remain stable within the time scale of routing decisions.}
We assume that the locations and frequency bands of nodes are shared, two nodes in a hop can achieve beam alignment, that is, they are located in one other main lobe of the beam. In the simulation, we employ the flap-top antenna model with a half-power beamwidth of 10 degrees. Since the beam alignment mechanism is introduced, the noise power exceeds the interference power by an order of magnitude. Therefore, we make a reasonable assumption that there is no co-frequency interference.

\begin{figure}[t!]
	\centering	\includegraphics[width=0.9\linewidth]{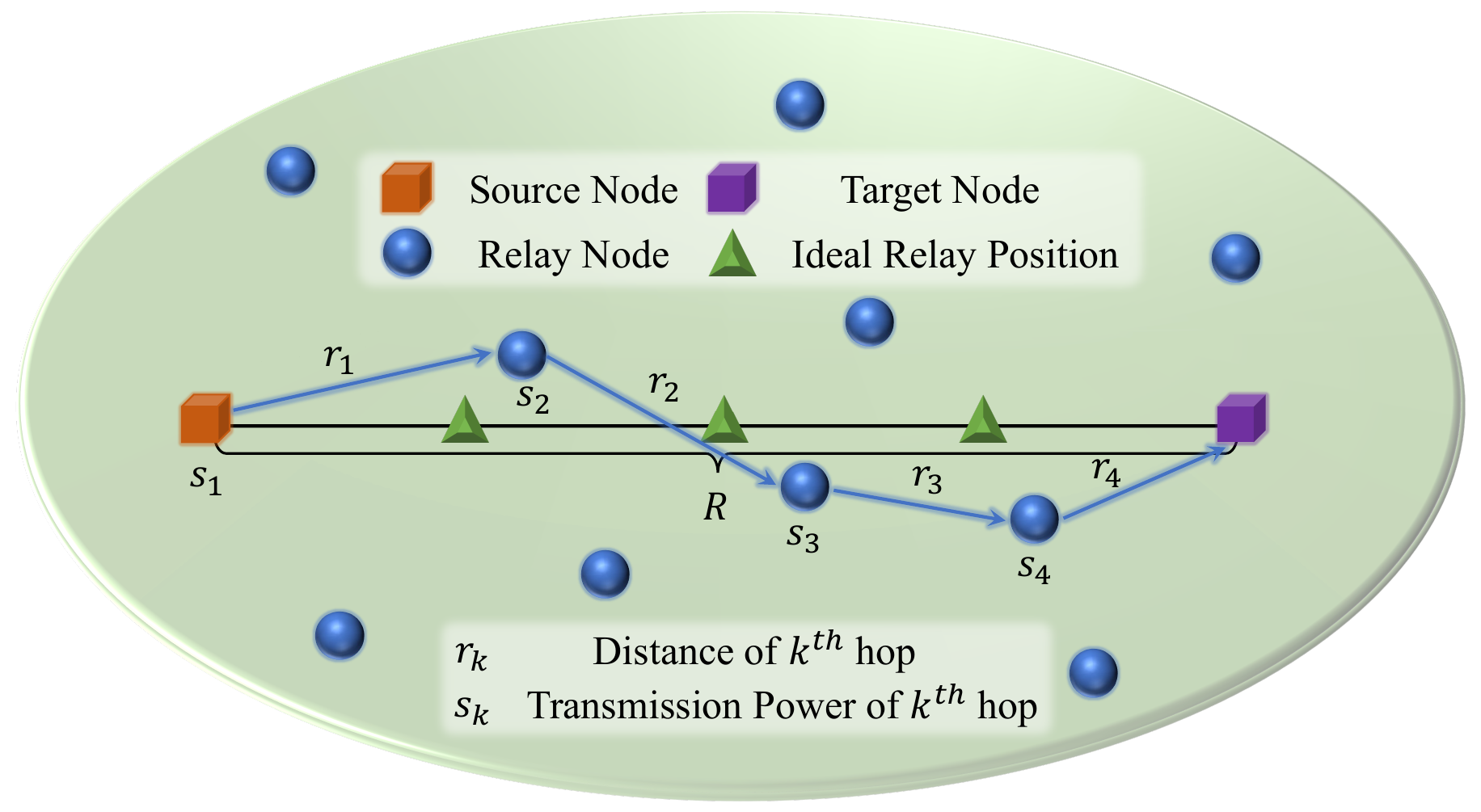}
	\caption{A scheme for maximum throughput routing.}
	\label{fig:sys}
	\vspace{-0.4cm}
\end{figure}

\subsection{THz Channel Model} 
We consider that the THz channel experiences large-scale fading and small-scale fading \cite{chaccour2019reliability}, while the \ac{SNR} of the THz channel is expressed as 
\begin{equation}\label{SNRTHz}
    {\rm{SNR}}_{{\rm{THz}}}  = \frac{S_{{\rm{THz}}} \,G_{{\rm{THz}}} \eta _{{\rm{THz}}} \, \exp\left({-\beta_{\rm{THz}} r}\right) \, \mathcal{X}_{\rm{THz}}}{\sigma_{\rm{THz}}^2\, r^2},
\end{equation}
where $S_{{\rm{THz}}}$ is the transmission power, $G_{{\rm{THz}}}$ is the antenna gain, and  $\sigma_{\rm{THz}}^2$ is the environmental noise which is in proportion to the transmission bandwidth $B_{\rm{THz}}$. The large-scale path loss is modeled as $\eta_{{\rm{THz}}} {\exp\left({-\beta_{\rm{THz}} r}\right)}/{r^2}$, where both mean addition losses $\eta _{\rm{THz}}$ and absorption coefficient $\beta_{\rm{THz}}$ are constants related to transmit frequency $\upsilon_{\rm{THz}}$ and propagation conditions \cite{olson2022coverage, jornet2011channel}.  Then, the small-scale fading $\mathcal{X}_{\rm{THz}}$ follows $\alpha-\mu$ distribution experienced with fading parameter $\alpha>0$ and normalized variance of the fading channel envelope $\mu>0$ \cite{lou2023coverage,bhardwaj2021fixed}, and is given by 
% \begin{equation}
%     f_{\mathcal{X}_{\rm{THz}}}(m) = \frac{\alpha \mu ^\mu m^{{\alpha \mu}-1}}{\overline{m}^{{\alpha \mu}}\Gamma (\mu)} \exp\left ( -\mu\left ( \frac{m}{\overline{m}} \right )^{{\alpha}} \right ) 
% \end{equation}
\begin{equation}
    f_{\mathcal{X}_{\rm{THz}}}(m) = \frac{\alpha \mu ^\mu m^{\frac{\alpha \mu}{2}-1}}{2\overline{m}^{\frac{\alpha \mu}{2}}\Gamma (\mu)} \exp\left ( -\mu\left ( \frac{m}{\overline{m}} \right )^{\frac{\alpha}{2}} \right ),
\end{equation}
where $ \Gamma (\mu) = \int_{0}^{\infty }t^{\mu-1}\exp{(-t)}{\rm d}t $ is theGamma function, and $\overline{m}  $ is the $\alpha$-root mean value of the fading channel envelope. The \ac{CCDF} of $\mathcal{X}_{\rm{THz}}$ is expressed as $\bar{F}_{\mathcal{X}_{\rm{THz}}}(m) = {\Gamma (\mu,(\mu (m/\overline{m} )^{\frac{\alpha}{2}})}/{\Gamma (\mu)}$, where ${\Gamma (\mu,\mu (m/\overline{m} )^{\frac{\alpha}{2}})}$ is the upper incomplete Gamma function.
%, as defined by Eq.(8.350.2) in \cite{gradshteyn2014table}.  

\vspace{-1mm}
\subsection{RF Channel Model}
The \ac{SNR} of the RF channel in one hop, defined as a combination of distance-dependent attenuation and small-scale components \cite{sayehvand2020interference}, is expressed as
\begin{equation}\label{SNRRF}
    {\rm{SNR}}_{{\rm{RF}}} = \frac{S_{{\rm{RF}}} \,G_{\rm{RF}} \eta _{\rm{RF}} \, r^{-\beta_{\rm{RF}}} \, \mathcal{X}_{\rm{RF}}}{\sigma_{\rm{RF}}^2},
\end{equation}
where $S_{{\rm{RF}}}$ is the transmission power, $G_{{\rm{RF}}}$ is antenna gain. Then, the large-scale fading of the RF channel is denoted as $\eta _{\rm{RF}} r^{-\beta_{\rm{RF}}}$, where $\eta _{\rm{RF}}$ is a mean addition loss related to transmit frequency $\upsilon_{\rm{RF}}$ and propagation conditions\cite{UAV_SG}, $r$ is the Euclidean distance between nodes in the hop, and $\beta_{\rm{RF}}$ is the path loss exponent for RF transmission. Moreover, $\sigma_{\rm{RF}}^2$ is the environmental noise, which is proportional to the transmission bandwidth $B_{\rm{RF}}$. Moreover, the $\mathcal{X}_{\rm{RF}}$ is small-scale fading followed exponential distribution with unit mean, i.e., the PDF of ${\mathcal{X}_{\rm{RF}}}$ and the CCDF of ${\mathcal{X}_{\rm{RF}}}$ follow same distribution,  that is, $f_{\mathcal{X}_{\rm{RF}}} (m)= \bar{F}_{\mathcal{X}_{\rm{RF}}}(m)= \exp{(-m)}$. 

{\color{black}
\subsection{Problem Statement}
In this subsection, we first introduce the definition of single-hop throughput and the relay communication mode. Based on that, the total throughput is provided and serves as the optimization objective. 
To ensure the tractability of the optimization process, we must appropriately handle the random variables involved in the optimization objective. 
We then present the optimization problem that aims to maximize the total throughput.

First, the single-hop throughput is defined as the ergodic capacity from the Shannon-Hartley theorem over a fading channel \cite{sayehvand2020interference}, which is given by
\begin{equation}\label{taoQk}
  \tau_{Q,k} \left(r_k,s_k\right) = B_{Q} \log_{2}\left(1+{\rm{SNR}}_{Q,k}\left(r_k,s_k\right) \right),   
\end{equation}
where $B_{Q}$ is the bandwidth of $Q$ link, $Q = \{\rm{THz}, \rm{RF}\}$. As illustrated in Fig.\ref{fig:sys}, $r_{k}$ and  $s_{k}$ are distances and transmission power of $k^{th}$ hops, respectively. 
{\color{black} 
In terms of communication mode, the relay nodes operate in half-duplex mode and adopt single-input single-output (SISO) configurations. Moreover, each relay node uses a directional antenna array to form a narrow beam and employs a beam alignment mechanism. 
Specifically, each relay node first aligns the main lobe of the beam with the previous node and receives the data packet in the buffering area. Then, it aligns the main lobe with the next node and transmits the packet.}
% As for the communication mode, since the beam alignment mechanism is applied, the relay node first aligns the main lobe of the beam with the previous node and receives the data packet in the buffering area. Then, it aligns the main lobe with the next node and transmits the packet. 
Therefore, the total buffering latency $\tau_Q^T$ is the sum of the buffering latency of each hop. The total throughput is given by the ratio of packet size to total buffering latency,
\begin{IEEEeqnarray}{RCL}\label{total tau}
    \tau_Q^T
    &=& \frac{M}{\sum_{k=1}^{K_{Q}} \frac{M}{\tau_{Q,k}\left(r_k,s_k\right)}} \notag \\
    &=& \frac{1}{\sum_{k=1}^{K_{Q}} \frac{1}{\tau_{Q,k}\left(r_k,s_k\right)}}, 
\end{IEEEeqnarray}
where $M$ is the size of the package, ${M}/{\tau_{Q,k}}$ is the buffering latency of the $k^{th}$ hop, and $\tau_{Q,k} \left(r_k,s_k\right)$ is defined in (\ref{taoQk}). Consider a large data package scenario in THz transmission, buffering latency is dominant compared to decoding latency and propagation latency.

{\color{black} To make the process of total throughput $\tau_Q^T$ maximization tractable, there are two random variables that need to be carefully handled. The first is the small-scale fading of the channel. In large-scale networks, acquiring channel state information (CSI) for all links incurs significant overhead. Since the instantaneous small-scale fading is not related to parameters such as transmission power or distance of each hop, we take the expectation of the expression in (\ref{SNRTHz}) or (\ref{SNRRF}) with respect to $\mathcal{X}_Q$, $Q = \{\rm{THz}, \rm{RF}\}$. Therefore, the average throughput for the $k^{th}$ hop in THz and RF networks are expressed as,
\begin{equation}\label{average tao}
  \overline{\tau}_{Q,k} \left(r_k,s_k\right) = B_Q \log_{2}\left(1+\overline{\rm{SNR}}_{Q,k}\left(r_k,s_k\right) \right),
\end{equation}
and 
\begin{equation}\label{SNR THz}
  \overline{\rm{SNR}}_{{\rm{THz}},k}\left(r_k,s_k\right) = \frac{s_{k} \,G_{{\rm{THz}}} \, \eta _{{\rm{THz}}}  \exp\left({-\beta_{\rm{THz}} r_k}\right)}{\sigma_{\rm{THz}}^2\, r_k^2} ,
\end{equation}
\begin{equation}\label{SNR RF}
  \overline{\rm{SNR}}_{{\rm{RF}},k}\left(r_k,s_k\right) =\frac{s_{k} \,G_{{\rm{RF}}} \, \eta _{{\rm{RF}}}  r^{-\beta_{\rm{RF}}} }{\sigma_{\rm{RF}}^2}.
\end{equation}}

{\color{black}The second random variable is the distance of each hop. Since the distribution of relay nodes obeys a PPP, the distance of one hop, that is, the distance between two nodes generally follows a distance distribution. This particular distance distribution will be discussed in subsection~\ref{single-hop distance}. However, taking account of distance distributions makes the results of maximum throughput routing extremely intricate. 
{\color{black} Therefore, we consider an ideal scenario in which there are no obstacles, and the selected relay nodes are all located on the segment between the source node and the target node. This ideal scenario closely mirrors vehicular communications, usually, the vehicles are on highway lanes where the inter-lane distance is very small, therefore, the source, relays, and target node are aligned. Moreover, the ideal relay positions correspond to the locations that maximize the throughput and serve as a theoretical upper bound for routing performance.} 
Since this segment is the shortest relay path, selecting the relay nodes on this line segment can shorten the distance of each hop, thus increasing throughput. The above simplification scheme can be expressed by the following constraint: $\sum_{k=1}^{{K_{Q}}} r_k  \geqslant  R.$} 
Finally, the optimization problem in $Q$ network is stated as follows
\vspace{-1mm}
\begin{subequations} 
	\begin{alignat}{2}
	\mathscr{P}_0: \; & \underset{ {{K_{Q}},r_1,\dots,r_{K_{Q}},s_1,\dots,s_{K_{Q}}}}{\textrm{maximize}}  &\; & \overline{\tau}_Q^T =  \frac{1}{\sum_{k=1}^{K_{Q}} \frac{1}{\overline{\tau}_{Q ,k}\left(r_k,s_k\right)}}, \label{opt0}\\
		&\quad \quad \; \textrm{subject to:}    &      & \sum_{k=1}^{{K_{Q}}} r_k  \geqslant  R, \label{st:constraint0-1}\\
		&    &   &   \sum_{k=1}^{{K_{Q}}} s_k \leqslant S_Q, \label{st:constraint0-2}
	\end{alignat}
\end{subequations}
where ${K_{Q}}$ is the number of hops, (\ref{st:constraint0-1}) and (\ref{st:constraint0-2}) respectively represent the limits of total communication distance $R$ and total transmission power $S_Q$. $R$ and $S_Q$ are predefined constants.
{\color{black} 
Noted that constraint (\ref{st:constraint0-2}) is imposed to reflect energy limitations in real-world systems, such as UAV-assisted networks or vehicular communications.}
% Unfortunately, the optimization objective $\overline{\tau}_Q^T$ is non-convex with respective to optimization variables. Therefore, we decompose the optimization process into three steps.
{\color{black}It is important to note that the optimization problem $\mathscr{P}_0$ is non-convex. 
Specifically, the average throughput of each hop, $\overline{\tau}_{Q,k}(r_k, s_k)$, is a logarithmic function of the $\overline{\textrm{SNR}}_{Q,k}\left(r_k,s_k\right)$, which itself is not jointly concave in $(r_k, s_k)$.
Moreover, hop number $K_Q$ is an integer variable. Therefore, the problem is a mixed-integer non-linear programming (MINLP) problem and cannot be solved using conventional convex optimization tools.
{\color{black} To address this challenge, we decompose the optimization process into three steps, which will be introduced in Sec.~\ref{section3}.}}}
% Following that, the analysis of relay nodes following a homogeneous PPP distribution will be presented in Sec.~\ref{section4}.}}

\section{Maximum Throughput Routing Strategies}\label{section3}
{\color{black}In this section, we adopt a stepwise optimization approach, in which each sub-problem (power allocation, relay selection, and number of hops design) is solved optimally to maximize the total throughput.}
% The power allocation, relay position selection, and number of hops design strategies are studied to maximize the total throughput. 
Note that the specific distribution of relay nodes is ignored in this section. 
Therefore, we have obtained the routing strategies that yield the highest throughput when the ideal relay positions are taken into account.
\vspace{-1mm}

\subsection{Power Allocation}
To solve $\mathscr{P}_0$, we first fix the number of hops ${K_{Q}}$ and distances of hops $r_1,r_2,\dots,r_{K_{Q}}$, only adjust the transmission power $s_1,s_2,\dots,s_{K_{Q}}$ of each hop to maximize the total throughput. 
Therefore, $\mathscr{P}_0$ is simplified into $\mathscr{P}_1^{Q}$, then
\vspace{-1mm}
\begin{subequations} 
	\begin{alignat}{2}
		\mathscr{P}_1^{Q}:\quad &\underset{ {s_1,s_2,\dots,s_{K_{Q}}}}{\text{maximize}}  &\quad& \overline{\tau}_Q^T = \frac{1}{\sum_{k=1}^{K_{Q}} \frac{1}{\overline{\tau}_{Q,k}\left(r_k,s_k\right)}}, \label{opt1}\\
		\; \; &\textrm{subject to:}    &\quad&   \sum_{k=1}^{{K_{Q}}} s_k \leqslant S_Q. \label{st:constraint1-1}
	\end{alignat}
\end{subequations}

Furthermore, to avoid transcendental equations caused by the logarithm term in (\ref{taoQk}), we focus on large SNR and small SNR cases. 
These two cases are sufficient to represent the majority of situations, because the value of SNR varies significantly with distance, especially for THz networks. 
The value of  SNR is often observed to be either $\gg 1$ or $\ll 1$, when solving the above sub-problem.  
Then, the power allocation strategy in THz and RF networks is given in the following proposition.

\begin{proposition}[Power Allocation Strategy] \label{prop1} 
Given the number of hops is ${K_{Q}}$ and the distances of hops are $r_1,r_2,\dots,r_{K_{Q}}$, the unique optimal power allocation strategy $s_1^*,s_2^*,\dots,s_{K_{Q}}^*$ of problem $\mathscr{P}_1^{Q}$ is,
\begin{IEEEeqnarray}{RCL}
    \left\{\!\!
 	 \begin{array}{ll} \!\!\!\!
  &\!\!\!\! \log_{2}\left({\overline{\rm{SNR}}_{Q,k}\left(r_k,s_k^*\right)}\right) \sqrt{s_k^*}  \\\!\!=\!&\!\! \!\!{\log_{2}\left({\overline{\rm{SNR}}_{Q,j}\left(r_j,s_j^*\right)}\right)} \sqrt{s_j^*}, \ \forall j,k,  \  {{\rm{for \ SNR}}\gg 1}\\\!\!&\!\!\!\!{\overline{\rm{SNR}}_{Q,k}\left(r_k,s_k^*\right)} \, s_k^*  \\\!\!=\!&\!\!\!\!{\overline{\rm{SNR}}_{Q,j}\left(r_j,s_j^*\right)} \, s_j^*, \ \forall j,k, 
 \ \ \ \ \ \ \  \ \ \  \; \, {{\rm{for \ SNR}}\ll 1}
	 \end{array} 
	\right.  ,  \IEEEeqnarraynumspace 
	%\vspace{-0.15cm}
\end{IEEEeqnarray}
where $\overline{\rm{SNR}}_{{\rm{THz}},k}\left(r_k,s_k\right)$ and $\overline{\rm{SNR}}_{{\rm{RF}},k}\left(r_k,s_k\right)$ are given in (\ref{SNR THz}) and (\ref{SNR RF}), respectively.
\end{proposition}
\begin{proof}
See Appendix~\ref{app:prop1}.
\end{proof}

% {\color{black} 
% From the above results, the following conclusions can be drawn. For a hop with a longer distance, it is suggested to allocate more power for this hop. 
% Importantly, we provide the corresponding quantitative derivation. 
% When SNR $\ll$ 1, the effect of the distance of hop $r_k$ on the power allocated to this hop $s_k^*$ is much greater than that of SNR $\gg$ 1.
% }

{\color{black} 
From the above results, the following conclusions can be drawn. For a hop with a longer distance, it is suggested to allocate more power for this hop. 
Importantly, we provide the corresponding quantitative derivation. 
As the link quality degrades, the effect of the distance of hop $r_k$ on the power allocated to this hop $s_k^*$ becomes more pronounced. 
}

\vspace{-2mm}
\subsection{Relay Position Selection}\label{subsection-3-3}
We consider that the number of hops ${K_{Q}}$ is fixed and the optimal power allocation strategy $s_1^*,s_2^*,\dots,s_{K_{Q}}^*$ is applied, we optimize $\overline{\tau}_Q^T$ by designing the distance of each hop. The corresponding optimization problem is $\mathscr{P}_2^Q$, such as 
\begin{subequations} 
	\begin{alignat}{2}
		\mathscr{P}_2^Q:\quad &\underset{ {r_1,r_2,\dots,r_{K_{Q}}}}{\textrm{maximize}}  &\quad& \overline{\tau}_Q^T = \frac{1}{\sum_{k=1}^{K_{Q}} \frac{1}{\overline{\tau}_{Q,k}\left(r_k,s_k^*\right)}}, \label{opt2}\\
		&\textrm{subject to:}    &      & \sum_{k=1}^{{K_{Q}}} r_k \geqslant R, \label{st:constraint2-1}
	\end{alignat}
\end{subequations}
where $R$ is the distance between the source node to the target node. The relay position selection strategy in THz and RF networks is shown in Proposition~\ref{prop2}.

\begin{proposition}[Relay Position Selection Strategy] \label{prop2} 
Given the number of hops is ${K_{Q}}$ and the optimal power allocation strategy in Proposition~\ref{prop1} is applied, the relay position selection strategy $r_1^*,r_2^*,\dots,r_{K_{Q}}^*$ of problem $\mathscr{P}_2^Q$ is,
\begin{IEEEeqnarray}{RCL}
    \left\{
 	\begin{array}{ll}
    { r_1^* = r_2^* = \dots = r_{K_{Q}}^* = \frac{R}{{K_{Q}}}}, & {{\rm{\ for \ SNR}}\gg 1} \\
     {\rm{No \ further \ requirements \ for \ }} & \\ \ \ \ \ \ \ \ \ \ \  \ \ \ \ \ r_1^*,r_2^*,\dots,r_{K_{Q}}^*, & {{\rm{\ for \ SNR}}\ll 1} 
	\end{array}
	\right. . \IEEEeqnarraynumspace 
	%\vspace{-0.15cm}
\end{IEEEeqnarray}
% where $Q = \{\rm{THz}, \rm{RF}\}$.
\begin{proof}
See Appendix~\ref{app:prop2}.
\end{proof}
\end{proposition}

Another interpretation of the above conclusion is as follows: In scenarios with a large SNR, it is advantageous to uniformly distribute the relay node locations. Conversely, in scenarios with a small SNR, maintaining equal distances between each hop is considered one of the optimal strategies. Interestingly, the findings obtained for non-uniform distribution often hold for the case of uniform distribution as well (as observed in subsection~\ref{number of hops}). Therefore, in maximum throughput routing, the relay positions are assumed to be uniformly distributed, even if SNR~$\ll$~1.

\subsection{Number of Hops Design}\label{number of hops}
Based on the power allocation and relay position selection strategies, the number of hops design is analyzed. This is the final step in deriving maximum throughput routing strategies. Under the condition that the relay position selection strategy is applied ($r_1^*=r_2^*=\dots=r_{K_{Q}}^*={R}/{{K_{Q}}}$),  $s_1^*=s_2^*=\dots=s_{K_{Q}}^*={S_Q}/{{K_{Q}}}$ is the unique solution for the power allocation strategy, as proved in Proposition~\ref{prop1}. Therefore, the specific optimization problem is $\mathscr{P}_3^Q$, 
\begin{equation} 
	% \begin{alignat}{2}
		\mathscr{P}_3^Q:\quad \underset{ {{K_{Q}}}}{\textrm{maximize}}  \quad \overline{\tau}_Q^T = \frac{\overline{\tau}_{Q,k}\left(\frac{R}{{K_{Q}}},\frac{S_Q}{{K_{Q}}}\right)}{{K_{Q}}}. \label{opt3}
	% \end{alignat}
\end{equation}
% where $Q = \{\rm{THz}, \rm{RF}\}$.

\begin{proposition}[Number of Hops Design] \label{prop3}
Given that the power allocation strategy described in Proposition~\ref{prop1} and the relay position selection strategy in Proposition~\ref{prop2} are applied, the optimal number of hops ${K_{Q}}^*$ in $\mathscr{P}_3^Q$ is given by:
\begin{itemize}
    \item For SNR $\ll$ 1, ${K_{Q}}^* \rightarrow + \infty$, that is, more hops bring to larger throughput.
    \item For SNR $\gg$ 1, ${K_{Q}}^*$ be obtained by rounding up and down the unique solution of the following transcendental equation for $\widetilde{K}_{Q}$,
\end{itemize}
\begin{IEEEeqnarray}{RCL}\label{optimal number of hops}
    \left\{
 	\begin{array}{ll}
    { \frac{2R\beta_{{\rm{THz}}}}{\widetilde{K}_{\rm{THz}}} - \ln \frac{S \,G_{{\rm{THz}}} \, \eta _{{\rm{THz}}}}{\sigma_{\rm{THz}}^2R^2} + 1 = \ln \widetilde{K}_{\rm{THz}} } \\
    { \widetilde{K}_{\rm{RF}} \left ( \frac{R}{\widetilde{K}_{\rm{RF}}} \right )^{\beta_{\rm{RF}}} = \exp\left ( 1- {\beta_{\rm{RF}}}+\ln \frac{S \,G_{{\rm{RF}}} \, \eta _{{\rm{RF}}}}{\sigma_{\rm{RF}}^2}\right ) }
	\end{array}
	\right. . \IEEEeqnarraynumspace 
	%\vspace{-0.15cm}
\end{IEEEeqnarray}
\begin{proof}
See Appendix~\ref{app:prop3}.
\end{proof}
\end{proposition}

Transcendental equations are unsolvable, however, since the hop number is an integer, low-complexity iterative algorithms can be employed to approximate the solutions. By rounding up and down, two potential optimal numbers of hops can be obtained, which require further comparison. Combining the above propositions, the maximum throughput routing strategies can be summarised in the following theorem.

\begin{theorem}[Ideal Routing Strategies]\label{theorem1}
The necessary and sufficient conditions for the maximum throughput strategy in both THz and RF networks are as follows.

\begin{itemize}
    \item When SNR $\ll$ 1, (\romannumeral1) the selected relay nodes are all located on the segment between the source node and target node, (\romannumeral2) the optimal number of hops ${K}_{Q}^*$ goes to infinity, and (\romannumeral3) the power is allocated according to ${{\overline{\rm{SNR}}_{Q,k}\left(r_k,s_k\right)} \, s_k = {\overline{\rm{SNR}}_{Q,j}\left(r_j,s_j\right)} \, s_j, \ \forall j,k}$ (the meaning of the variables is described in Proposition~\ref{prop1}).
    \item When SNR $\gg$ 1, (\romannumeral1) the selected relay nodes are all located on the segment between the source node and target node, (\romannumeral2) the optimal number of hops ${K}_{Q}^*$ is obtained by rounding up and down the unique solution of the transcendental equation defined in (\ref{optimal number of hops}), and (\romannumeral3) each hop has the same distance and transmission power.
\end{itemize}
\end{theorem}

\section{Throughput and Coverage Analysis Under Maximum Throughput Routing}\label{section4}
{\color{black}Since the spatial distribution of relay nodes was not considered in Section~\ref{section3}, the throughput associated with the relay locations determined by the ideal routing strategy (Theorem~1) serves as an ideal upper bound. In this section, we focus on the throughput based on homogeneous PPP distributions of relay nodes.}
% Since the relay nodes distribution was not considered in Section~\ref{section3}, the throughput associated with the relay locations determined by Theorem~\ref{theorem1} represents an ideal upper bound. In this section, we focus on the throughput based on specific homogeneous PPP distributions. 
Additionally, we explore another well-known metric, namely coverage probability, which can also be interpreted as an expression of instantaneous throughput at low SNR.

\vspace{-2mm}
\subsection{Single-hop Distance Distributions} \label{single-hop distance}
Before we start deriving the analysis of the key metrics, we need three tradeoffs to transform the ideal maximum throughput strategy from the ideal positions given in Theorem~\ref{theorem1} to the specific PPP-based relay distribution. The reasons behind these tradeoffs and their corresponding solutions are as follows.

\begin{itemize}
    \item Since it is impractical to find relay nodes at exactly the ideal positions, we instead look for the nearest nodes of ideal positions as relay nodes.
    \item Due to the randomness of relay nodes' positions, the distance of each hop in the routing is different from others. For the power allocation strategy, determining the power of each hop requires knowing the distance of all hops. Based on the uniform distribution assumption in Subsection~\ref{subsection-3-3}, the specific distance difference between each hop is ignored and each hop is assigned with the same power.
    \item When SNR $\ll$ 1, the optimal  number of hops ${K}_{Q}^* \rightarrow \infty$ in the ideal scenario. However, in the case where the number of relay nodes is limited, an excessive number of hops will not only cause the same relay to be selected by adjacent hops but also reduce the total throughput. Considering the null probability property of PPP, hop number selection should adhere to the following criteria,

    \begin{equation}\label{criteria}
        K_Q < \sqrt{\frac{\lambda_Q \pi R^2}{4\ln\frac{1}{\varepsilon}} },
    \end{equation}
    where $Q=\{ {\rm{THz}}, {\rm{RF}} \}$. When (\ref{criteria}) is satisfied, the probability that the two adjacent hops select the same relay is not greater than $1-\left( 1 - \varepsilon \right)^{K_Q}$.
\end{itemize}

Based on the first tradeoff, the distance of each hop is characterized as a random variable. This random variable can be categorized into two types, which are defined as follows.

\begin{definition}[Type-\uppercase\expandafter{\romannumeral1} Distance Distribution]
Assume $x_1$ and $x_2$ are reference points (do not belong to the PPP) with distance $r$. Denote the nearest point in PPP to $x_2$ is $y_2$. The distance distribution $f_{Q}^{(1)} (\rho | r)$ between $x_1$ and $y_2$ is called the type-\uppercase\expandafter{\romannumeral1} distance distribution, where $Q=\{ \rm{THz,RF} \}$ and $\rho$ is the distance between $x_1$ and $y_2$.
\end{definition}

Since the source node and target node in the routing remain fixed, the distance of the first hop and the distance of the last hop follow the type-\uppercase\expandafter{\romannumeral1} distance distribution. The \ac{PDF} of type-\uppercase\expandafter{\romannumeral1} distance distribution can be expressed as follows.

\begin{lemma}[PDF of Type-\uppercase\expandafter{\romannumeral1} Distance Distribution] \label{lemma1}
Given that the distance of two reference points is $r$, the PDF of type-\uppercase\expandafter{\romannumeral1} distance distribution is,
\begin{equation}\label{PDF_1}
\begin{split}
    f_{Q}^{(1)} (\rho | r) =\int_{r-\rho}^{r+\rho} \frac{2 \lambda_Q \rho \exp \left( - \lambda_Q \pi l^2 \right)} {r \sqrt{ 1 - \frac{\left(r^2 + l^2 - \rho^2  \right)^2 }{4 r^2 l^2} }} \mathrm{d} l.
\end{split}
\end{equation}
% where $Q=\{ \rm{THz,RF} \}$.
\begin{proof}
See Appendix~\ref{app:lemma1}. 
\end{proof}
\end{lemma}

\begin{definition}[Type-\uppercase\expandafter{\romannumeral2} Distance Distribution]
Assume $x_1$ and $x_2$ are reference points (do not belong to the PPP) with distance $r$. Denote the nearest points in PPP to $x_1$ and $x_2$ are $y_1$ and $y_2$ respectively. The distance distribution $f_{Q}^{(2)} (\rho | r)$ between $y_1$ and $y_2$ is called the type-\uppercase\expandafter{\romannumeral2} distance distribution, where $\rho$ is the distance between $y_1$ and $y_2$.
\end{definition}

As opposed to the type-\uppercase\expandafter{\romannumeral1} distance distribution, type-\uppercase\expandafter{\romannumeral2} distance distribution describes the distance distribution of the middle hops (except the first hop and the last hop). The \ac{PDF} of type-\uppercase\expandafter{\romannumeral2} distance distribution is derived from the PDF of type-\uppercase\expandafter{\romannumeral1} distance distribution.

\begin{lemma} [PDF of Type-\uppercase\expandafter{\romannumeral2} Distance Distribution] \label{lemma2}
Given that the distance of two reference points is $r$, the PDF of type-\uppercase\expandafter{\romannumeral2} distance distribution is,
\begin{IEEEeqnarray}{RCL}\label{PDF_2}
    f_{Q}^{(2)} (\rho | r)  &= & \int_{0}^{\infty }\int_{0}^{2\pi}\lambda _{Q}  \exp\left ( - \lambda _{Q} \pi {\hat{\rho}}^2 \right )  \notag \\   
   & & \times \int_{\widehat{r}-\rho}^{\widehat{r}+\rho} \frac{2 \lambda_Q \rho \exp \left( - \lambda_Q \pi l^2 \right) } {\widehat{r} \sqrt{ 1 - \frac{\left(\widehat{r}^2 + l^2 - \rho^2  \right)^2 }{4 \widehat{r}^2 l^2} }} \mathrm{d} l \mathrm{d} \theta \mathrm{d} \hat{\rho},
\end{IEEEeqnarray}
where $\widehat{r} =\sqrt{r^2+\hat{\rho}^2-2r \hat{\rho} \cos\theta}$.
\begin{proof}
See Appendix~\ref{app:lemma2}. 
\end{proof}
\end{lemma}

\subsection{Throughput Analysis}
In this subsection, we first propose a stepwise-optimal maximum throughput routing strategy. Then, we introduce a stepwise-suboptimal routing strategy, which is more suitable for analysis, according to the tradeoff schemes.

% In this subsection, we first propose a suboptimal maximum throughput strategy according to the tradeoff schemes. This strategy is more suitable for analysis and based on a specific PPP.

\begin{definition}[Stepwise-Optimal Routing Strategy]\label{def3}
For both THz and RF networks, the stepwise-optimal maximum throughput routing strategies satisfy the following three conditions:
\begin{itemize}
    \item The optimal number of hops $\widehat{K}_Q$ (where $Q= \{ \rm{THz,RF} \}$, the same below), is first derived by Proposition~\ref{prop3}. If the criteria in (\ref{criteria}) is not satisfied, $\widehat{K}_Q$ is replaced as the largest integer that holds the inequation (\ref{criteria}).
    \item The $\widehat{K}_Q - 1$ ideal relay positions are uniformly distributed on the segment connecting the source node and the target node. The $\widehat{K}_Q - 1$ relay nodes are the nearest nodes of ideal relay positions.
    \item The transmission power of each hop is allocated according to the power allocation strategy given in Proposition~\ref{prop1}.
\end{itemize}
\end{definition}

Furthermore, to facilitate the derivation within the SG framework, we have made a compromise and introduced a stepwise-suboptimal strategy.

\begin{definition}[Stepwise-Suboptimal Routing Strategy]\label{def4}
In both THz and RF networks, the stepwise-suboptimal maximum throughput routing strategies align with the stepwise-optimal routing strategy, except for the third condition: in stepwise-suboptimal routing strategies, the transmission power of each hop is $\frac{S}{\widehat{K}_Q}$ 
\end{definition}

Based on the above two lemmas and Definition~\ref{def4}, the analytical expression of throughput is given by the following theorems.

\begin{theorem}[Throughput of THz Networks]\label{theorem2}
Given that the stepwise-suboptimal maximum throughput routing strategy is applied, the throughput of THz networks, denoted as $\tau_{\rm{THz}}^T$, is given by,
\begin{IEEEeqnarray}{RCL}
    \tau_{\rm{THz}}^T & =& \frac{1}{\frac{2}{\tau_{{\rm{THz}},1}}+\frac{K_{\rm{THz}}-2}{\tau_{{\rm{THz}},2}}}  \\ &=&\frac{{\tau_{{\rm{THz}},1}}\cdot {\tau_{{\rm{THz}},2}}}{2{\tau_{{\rm{THz}},2}}+\left( K_{\rm{THz}}-2 \right){\tau_{{\rm{THz}},1}}},
\end{IEEEeqnarray}
where ${\tau_{{\rm{THz}},k}}$ is expressed as,
\begin{IEEEeqnarray}{RCL}
    &\tau_{{\rm THz},k}& \!=\!\! \int_{0}^{\infty }\int_{0}^{\infty} \frac{f_{{\rm{THz}}}^{(k)}\left(\rho \big| \frac{R}{K_{\rm{THz}}} \right)}{\Gamma\left(\mu\right)}  \notag \\
    &\times \! \Gamma \!\Bigg( \!\mu,&\! \mu \! \left( \! \frac{ K (2^{\frac{t}{B_{\rm THz}}}\!-\!1) \rho^2 \!\exp\left(\beta_{\rm{THz}} \rho \right) \sigma_{\rm{THz}}^2 }  {S \cdot G_{{\rm{THz}}} \eta_{{\rm{THz}}}} \right)^{\!\frac{\alpha}{2}} \! \Bigg) {\rm d} \rho {\rm d} t,  \IEEEeqnarraynumspace 
\end{IEEEeqnarray}
where $k=1,2$, $f_{{\rm{THz}}}^{(1)}\left(\rho \big| \frac{R}{K_{\rm{THz}}} \right)$ and $f_{{\rm{THz}}}^{(2)}\left(\rho \big| \frac{R}{K_{\rm{THz}}} \right)$ are defined in (\ref{PDF_1}) and (\ref{PDF_2}), respectively.
\begin{proof}
See Appendix~\ref{app:theorem2}.
\end{proof}
\end{theorem}

\begin{theorem}[Throughput of RF Networks]\label{theorem3}
Given that the stepwise-suboptimal maximum throughput routing strategy is applied, the throughput of RF networks, denoted as $\tau_{\rm{RF}}^T$, is given by,
\begin{IEEEeqnarray}{RCL}
    \tau_{\rm{RF}}^T &=& \frac{1}{\frac{2}{\tau_{{\rm{RF}},1}}+\frac{K_{\rm{RF}}-2}{\tau_{{\rm{RF}},2}}} \\
    &=& \frac{{\tau_{{\rm{RF}},1}}\cdot {\tau_{{\rm{RF}},2}}}{2{\tau_{{\rm{RF}},2}}+\left( K_{\rm{RF}}-2 \right){\tau_{{\rm{RF}},1}}},
\end{IEEEeqnarray}
where ${\tau_{{\rm{RF}},k}}$ is expressed as,
\begin{IEEEeqnarray}{RCL}
    &\tau_{{\rm RF},k} =& \int_{0}^{\infty }\int_{0}^{\infty} f_{{\rm{RF}}}^{(k)}\left(\rho \bigg| \frac{R}{K_{\rm{RF}}} \right)  \notag \\
    &\times  \exp & \left(  -{(2^{\frac{t}{B_{\rm RF}}}-1) \frac{K}{S} \rho^{\beta_{\rm{RF}}}  \sigma_{\rm{RF}}^2 }  { G_{{\rm{RF}}}^{-1} \eta_{{\rm{RF}}}^{-1}} \right) {\rm d} \rho {\rm d} t, \IEEEeqnarraynumspace 
\end{IEEEeqnarray}
where $k=1,2$, $f_{{\rm{RF}}}^{(1)}\left(\rho \big| \frac{R}{K_{\rm{RF}}} \right)$ and $f_{{\rm{RF}}}^{(2)}\left(\rho \big| \frac{R}{K_{\rm{RF}}} \right)$ are defined in (\ref{PDF_1}) and (\ref{PDF_2}), respectively.
\begin{proof}
See Appendix~\ref{app:theorem3}.
\end{proof}
\end{theorem}

\begin{table*}
\vspace{-4mm}
\begin{IEEEeqnarray}{RCL} \label{CP of THz}
    P_{\rm{THz}}^C &=& \left ( \int_{0}^{\infty} \frac{f_{{\rm{THz}}}^{(1)}\left(\rho \big| \frac{R}{K_{\rm{THz}}} \right)}{\Gamma\left(\mu\right)} {\Gamma \left(\mu, \mu \, \left(  \frac{K_{\rm{THz}} \gamma_{\rm{THz}} \rho^2 \exp\left(\beta_{\rm{THz}} \rho \right) \sigma_{\rm{THz}}^2 }  {S G_{{\rm{THz}}} \eta_{{\rm{THz}}}} \right)^{\frac{\alpha}{2}} \right)} {\rm d} \rho\right )^2 \notag \\
    & & \times \left ( \int_{0}^{\infty} \frac{f_{{\rm{THz}}}^{(2)}\left(\rho \big| \frac{R}{K_{\rm{THz}}} \right)}{\Gamma\left(\mu\right)} {\Gamma \left(\mu, \mu \, \left(  \frac{K_{\rm{THz}} \gamma_{\rm{THz}} \rho^2 \exp\left(\beta_{\rm{THz}} \rho \right) \sigma_{\rm{THz}}^2 }  {S G_{{\rm{THz}}} \eta_{{\rm{THz}}}} \right)^{\frac{\alpha}{2}} \right)} {\rm d} \rho\right )^{K_{\rm{THz}}-2}.
\end{IEEEeqnarray}

\begin{IEEEeqnarray}{RCL} \label{CP of RF}
       P_{\rm{RF}}^C  = \left ( \int_{0}^{\infty}  \!\!\!f_{\rm{RF}}^{(1)}\left(\! \rho \bigg| \frac{R}{K_{\rm{RF}}} \! \right)  \! \exp\left(\! - \!\frac{K_{\rm{RF}}}{S} {\gamma_{\rm{RF}} \rho^{\beta_{\rm{RF}}}  \sigma_{\rm{RF}}^2 }  {G_{{\rm{RF}}}^{-1} \eta_{{\rm{RF}}}^{-1}} \right) \! {\rm d} \rho \! \right )^{\!2} \!  \left ( \int_{0}^{\infty} \!\!\! f_{\rm{RF}}^{(2)}\!\left(\!\rho \bigg| \frac{R}{K_{\rm{RF}}}\!\right) \!\exp\left(\! -\! \frac{K_{\rm{RF}}}{S} {\gamma_{\rm{RF}} \rho^{\beta_{\rm{RF}}}  \sigma_{\rm{RF}}^2 }  {G_{{\rm{RF}}}^{-1} \eta_{{\rm{RF}}}^{-1}} \right) \! {\rm d} \rho \! \right )^{\! K_{\rm{RF}}-2}. \IEEEeqnarraynumspace 
\end{IEEEeqnarray}
\noindent\rule{\linewidth}{0.2mm}
\vspace{-6mm}
\end{table*}

\vspace{-3mm}

\subsection{Coverage Analysis Under the Low SNR}
The coverage probability represents the probability that every hop in the multi-hop routing can maintain a stable connection, which can be mathematically expressed as $\mathbb{P} [{\rm{SNR}} > \gamma_Q], \ \forall k \leq K_Q$, where $Q=\{ \rm{THz}, \rm{RF}\}$. $\gamma_Q$ is called the SNR threshold, which is the threshold of being connected. The coverage probabilities of the THz and RF networks are respectively provided in the following two theorems.

\begin{theorem}[Coverage Probability of THz Networks]\label{theorem4}
Given that the stepwise-suboptimal maximum throughput routing strategy is applied, the coverage probability of THz networks, denoted as $P_{\rm{THz}}^C$, is given by (\ref{CP of THz}), where $k=1,2$, $f_{\rm{THz}}^{(1)}\left(\rho \big| \frac{R}{K_{\rm{THz}}} \right)$ and $f_{\rm{THz}}^{(2)}\left(\rho \big| \frac{R}{K_{\rm{THz}}} \right)$ are defined in (\ref{PDF_1}) and (\ref{PDF_2}), respectively.
\begin{proof}
See Appendix~\ref{app:theorem4}.
\end{proof}
\end{theorem}

\begin{theorem}[Coverage Probability of THz Networks]\label{theorem5}
Given that the stepwise-suboptimal maximum throughput routing strategy is applied, the coverage probability of THz networks, denoted as $P_{\rm{RF}}^C$, is given by (\ref{CP of RF}),
% \begin{IEEEeqnarray}{RCL}
%       & P_{\rm{RF}}^C &  = \left ( \int_{0}^{\infty} f_{\rm{RF}}^{(1)}\left(\rho \bigg| \frac{R}{K_{\rm{RF}}} \right)   \exp\left( - \frac{K_{\rm{RF}}}{S} {\gamma_{\rm{RF}} \rho^{\beta_{\rm{RF}}}  \sigma_{\rm{RF}}^2 }  {G_{{\rm{RF}}}^{-1} \eta_{{\rm{RF}}}^{-1}} \right) {\rm d} \rho\right )^2 \notag \\
%       & \times& \!\!\! \!\left ( \int_{0}^{\infty} \!\!\! f_{\rm{RF}}^{(2)}\!\left(\!\rho \bigg| \frac{R}{K_{\rm{RF}}}\!\right) \!\exp\left( - \frac{K_{\rm{RF}}}{S} {\gamma_{\rm{RF}} \rho^{\beta_{\rm{RF}}}  \sigma_{\rm{RF}}^2 }  {G_{{\rm{RF}}}^{-1} \eta_{{\rm{RF}}}^{-1}} \right) {\rm d} \rho \right )^{K_{\rm{RF}}-2}
% \end{IEEEeqnarray}
where $f_{\rm{RF}}^{(1)}\left(\rho \big| \frac{R}{K_{\rm{RF}}} \right)$ and $f_{\rm{RF}}^{(2)}\left(\rho \big| \frac{R}{K_{\rm{RF}}} \right)$ are defined in (\ref{PDF_1}) and (\ref{PDF_2}), respectively.
\begin{proof}
The proof of Theorem~\ref{theorem5} is similar to that of Theorem~\ref{theorem3} and Theorem~\ref{theorem4}, therefore omit here.
\end{proof}
\end{theorem}

Since a high throughput guarantees a stable connection, it is crucial to analyze the coverage probability under low SNR. 
{\color{black} In general, we assume that $\gamma_{\rm{THz}} B_{\rm{THz}} = \gamma_{\rm{RF}} B_{\rm{RF}}$ to ensure a fair comparison between THz and RF networks. Both SNR thresholds and bandwidths are constant system parameters and do not vary with the positions of the relay nodes.}
% In general, we assume that $\gamma_{\rm{THz}} B_{\rm{THz}} = \gamma_{\rm{RF}} B_{\rm{RF}}$. 
There are two justifications for this assumption. Firstly, according to the Shannon-Hartley theorem under low SNR, $\tau = B \log_2 ( {1+{\rm{SNR}}}) \approx \frac{B}{\ln 2} {\rm{SNR}}$. Consequently, a larger bandwidth allows a lower SNR threshold due to the increased capacity of the channel.
Secondly, the THz signal can be obtained by spreading the RF signal spectrum at the source node, and only the noise power within the RF band remains after the spread spectrum decomposition at the target node. 
As a result, the received noise power of THz becomes $\frac{B_{\rm{RF}}}{B_{\rm{THz}}}$ times that of the RF, which leads to the SNR threshold of THz becoming $\frac{B_{\rm{THz}}}{B_{\rm{RF}}}$ times that of the RF.

\vspace{-0.1cm}
\section{Numerical Results} \label{section5}
In this section, we validate theorems and propositions using Monte Carlo simulations and demonstrate the advantages of the proposed routing strategies. Moreover, we contrast THz and RF transmissions and integrate multi-hop THz routing with UAV communications. Additionally, the values of parameters are given in Table~\ref{table1}, unless otherwise specified.
%\cite{sayehvand2020interference}

\begin{table}[t]
\centering
\caption{{\color{black}Summary of Parameters \cite{gordon2017hitran2016, lou2023coverage}.}}
\label{table1}
\resizebox{\linewidth}{!}{
 \renewcommand{\arraystretch}{1.2}
\begin{tabular}{|c|c|c|}
\hline
Description      & Parameter        & Default Value      \\ \hline \hline
Antenna gain of THz/RF           &  $G_{\rm{THz}}$, $G_{\rm{RF}}$        & $20$, $0~$[dBi]   \\ \hline
Transmit frequency of THz/RF     & $\upsilon_{\rm{THz}}$, $\upsilon_{\rm{RF}}$ & $1~$THz, $2.1~$GHz \\ \hline
Mean additional losses of THz/RF  &$\eta_{\rm{THz}}$,  $\eta_{\rm{RF}}$        & $-93$, $-39$~[dB]   \\ \hline
Absorption coefficient of THz & $\beta _{\rm{THz}}$ & 0.05~/m \\ \hline
Path loss exponent of RF   & $\beta_{\rm{RF}}$   & 2.5        \\ \hline
Bandwidth of THz/RF & $B_{\rm{THz}}$, $B_{\rm{RF}}$          & $500$, $40$~[MHz]  \\ \hline
Environmental noise of THz/RF  & $\sigma_{\rm{THz}}^2$, $\sigma_{\rm{RF}}^2$ & $-107$, $-128$~[dBm]  \\ \hline
Parameters of $\alpha-\mu$ distribution & $\alpha$, $\mu$     & 2, 4       \\ \hline
Density of THz/RF relay nodes & $\lambda_{\rm{THz}}$, $\lambda_{\rm{RF}}$    &   $10$, $0.5$~[$\cdot 10^{-3}/\rm{m}^{2}$]   \\ \hline
SNR threshold of THz/RF & $\gamma_{\rm{THz}}$, $\gamma_{\rm{RF}}$    &   $0$, $0$~[dB]   \\ \hline
\end{tabular}}
\vspace{-0.4cm}
\end{table}

\vspace{-2mm}
\subsection{Strategies Comparison}
In this subsection, we introduce five routing strategies presented in Fig.~\ref{fig1} and Fig.~\ref{fig2}. Subsequently, we verify the accuracy of theorems and propositions through Monte Carlo simulations. Finally, we compare the throughput performances of different strategies.

\par
In this paper, we propose three routing strategies: ideal (Theorem~\ref{theorem1}), stepwise-optimal (Definition~\ref{def3}), and stepwise-suboptimal (Definition~\ref{def4}) maximum throughput routing strategies. 
{\color{black}To evaluate the performance of the proposed routing strategies, we compare them with existing SG-based routing strategies, most of which fall into either long-hop \cite{wang2023reliability, wang2022stochastic, farooq2015stochastic} or short-hop \cite{sasaki2017energy, farooq2015stochastic, banaei2014asymptotic} categories.
In the long-hop routing strategy, the source or relay node selects the relay node that is closest to the destination within a fixed radius, which is set to $40$~m for THz and $400$~m for RF. If the destination falls within this radius, it is selected directly.
Conversely, the short-hop routing strategy involves either the source or a relay node selecting the nearest node as the next-hop relay node, provided that the node's direction angle is less than ${\pi}/{4}$. The direction angle is defined by considering the relay node (or source node) of the previous hop as the common vertex, and creating two rays: one extending from this vertex to the next-hop relay node, and another from the vertex to the target node. This criterion ensures that routing does not deviate significantly from the intended direction.}

\begin{figure*}[ht!]
\vspace{-5mm}
\begin{minipage}[t]{0.32\linewidth}
\centering
\includegraphics[width=\linewidth]{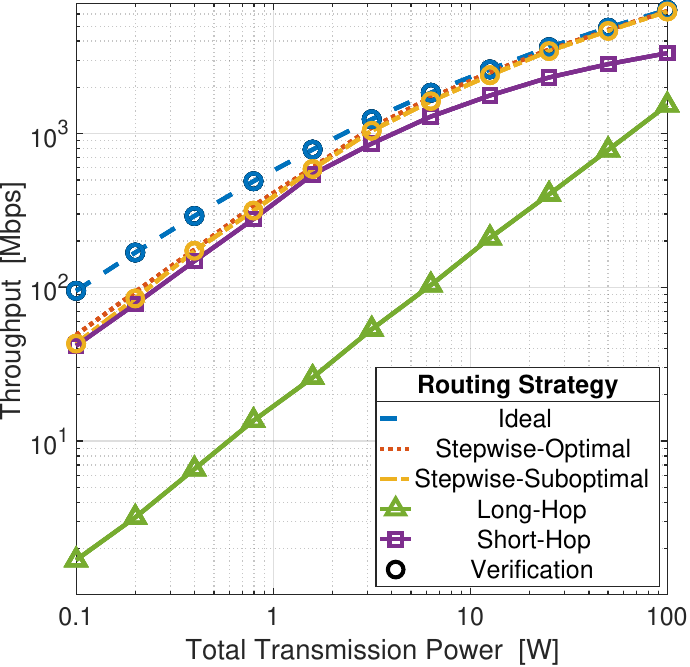}
\caption{Strategies comparison for THz transmission.}
\label{fig1}
\end{minipage}
\hfill
\begin{minipage}[t]{0.32\linewidth}
\centering
\includegraphics[width=\linewidth]{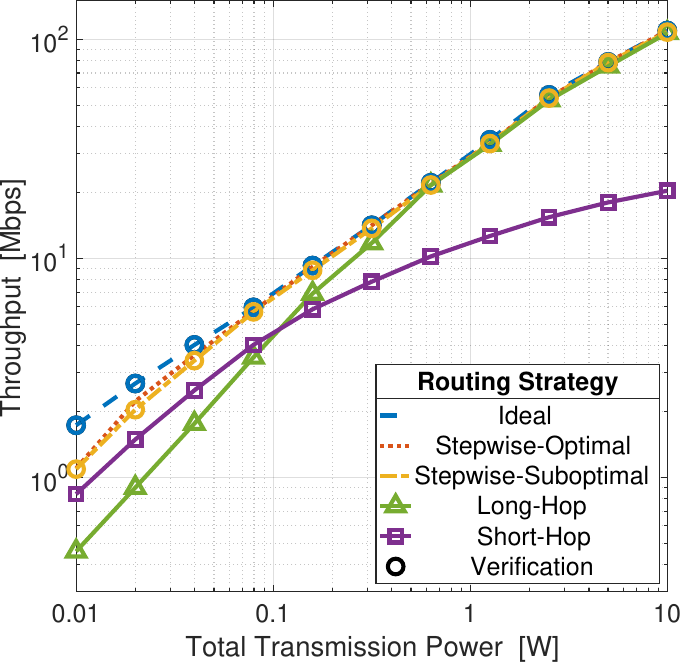}
\caption{Strategies comparison for RF transmission.}
\label{fig2}
\end{minipage}
\hfill
\begin{minipage}[t]{0.32\linewidth}
\centering
\includegraphics[width=\linewidth]{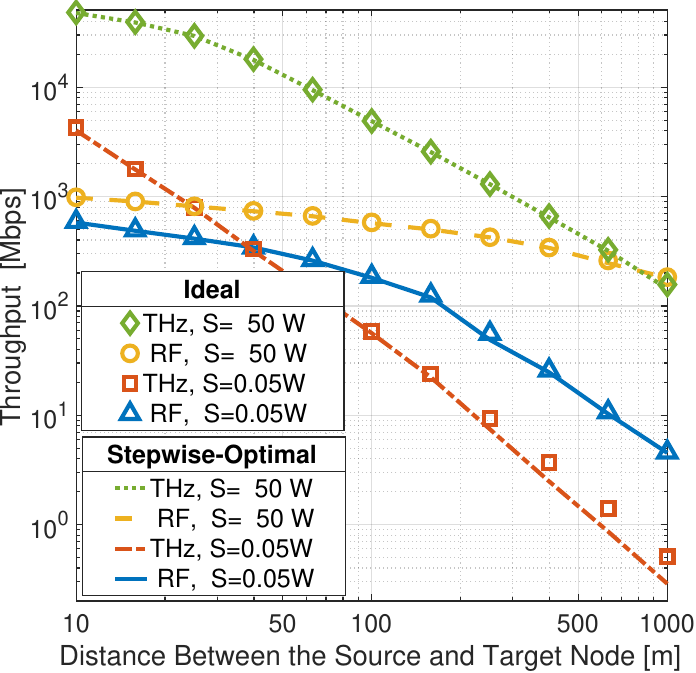}
\caption{Throughput Performance Comparison of THz and RF Transmission.}
\label{fig3}
\end{minipage}
\vspace{-5mm}
\end{figure*}

% \vspace{-2mm}

\par
To ensure comprehensive verification, we incorporate a wide range of transmission power levels, encompassing both low and high SNR conditions. 
The distances from the source node to the target node are $100$~m for the THz network and $1000$~m for the RF network. 
In Fig.~\ref{fig1} and Fig.~\ref{fig2}, the blue dashed lines are drawn according to Theorem.~\ref{theorem1}. For verification, we first fixed the number of hops, randomly generated the distance and transmission power of each hop, and measured the maximum average multi-hop throughput. Next, the optimal hop is determined by exhaustive search, and the maximum throughput with the optimal hop is marked as a blue circle and shown in Fig.~\ref{fig1} and Fig.~\ref{fig2}. The matching of blue circles and blue dashed lines proves the accuracy of the results in Section~\ref{section3}. Moreover, the yellow dash-dot lines are drawn according to the analytical expressions in Theorem~\ref{theorem2} and Theorem~\ref{theorem3}, whereas the yellow circles are obtained by simulation according to the description given by Definition.~\ref{def4}. The matching of yellow circles and yellow dash-dot lines proves the accuracy of Theorem~\ref{theorem2} and Theorem~\ref{theorem3}. Furthermore, the analytical expressions in Theorem~\ref{theorem4} and Theorem~\ref{theorem5} are also validated using this approach.

\par
{\color{black}
Then, we compare the throughput of different strategies through Fig.~\ref{fig1} and Fig.~\ref{fig2}. 
As mentioned before, the ideal routing strategy serves as an upper bound for throughput. In most cases, the proposed stepwise-optimal routing strategy can approximate the ideal routing strategy in terms of throughput performance. At lower transmission powers, the ideal routing strategy has a large or potentially infinite number of hops, which is theoretically but not practically feasible. Under the restriction of criterion (\ref{criteria}), the stepwise-optimal routing strategy exhibits a minor difference in throughput compared to the ideal routing strategy.
Furthermore, the stepwise-suboptimal routing strategy can provide a tight lower bound to the stepwise-optimal routing strategy. Because the stepwise-suboptimal routing strategy is tractable, it can be a viable alternative to the stepwise-optimal routing strategy.  

\par
Notably, in THz communication, there is a bias toward the short-hop routing strategy, whereas, in RF communication, the short-hop routing strategy exhibits distinct advantages over the long-hop routing strategy primarily at lower transmission power levels. It is due to the logarithmic relationship of throughput with SNR, which leads to diminishing marginal returns on throughput as power increases. Overall, as demonstrated in  Fig.~\ref{fig1} and Fig.~\ref{fig2}, the ideal, stepwise-optimal, and stepwise-suboptimal routing strategies outperform the baseline strategies, i.e., the long-hop and short-hop routing strategies. This indicates that our strategies effectively optimize throughput through number of hops design, relay selection, and, importantly, power allocation, which is not considered in existing SG-based routing strategies.} 

\begin{figure*}[tbp]
\vspace{-6mm}
\begin{minipage}[t]{0.29\linewidth}
\centering
\includegraphics[width=\linewidth]{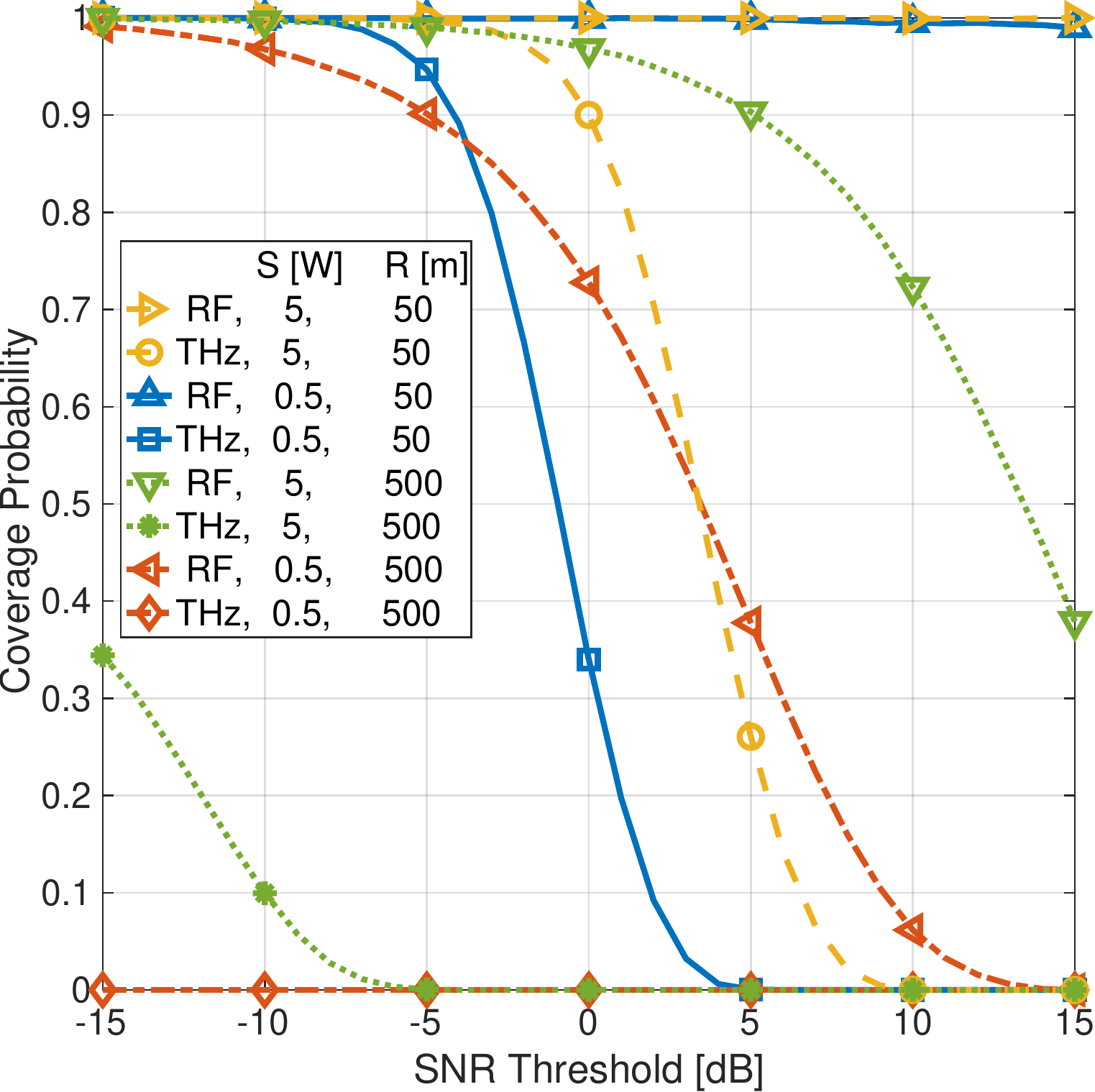}
\caption{Coverage Performance Comparison of THz and RF Transmission.}
\label{fig4}
\end{minipage}
\hfill
\begin{minipage}[t]{0.345\linewidth}
\centering
\includegraphics[width=\linewidth]{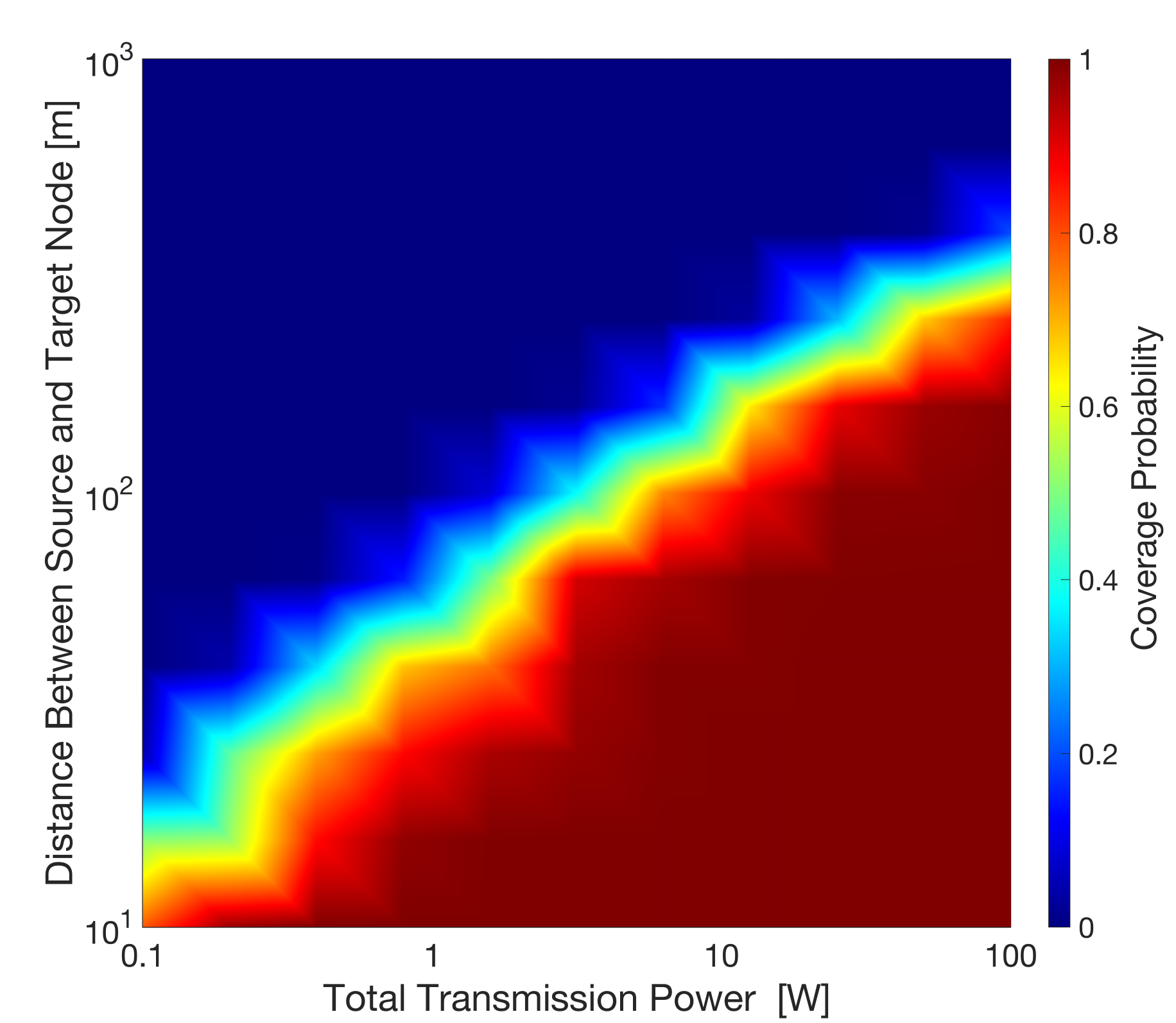}
\caption{Heat Map of Coverage Probability for THz Transmission.}
\label{fig5}
\end{minipage}
\hfill
\begin{minipage}[t]{0.345\linewidth}
\centering
\includegraphics[width=\linewidth]{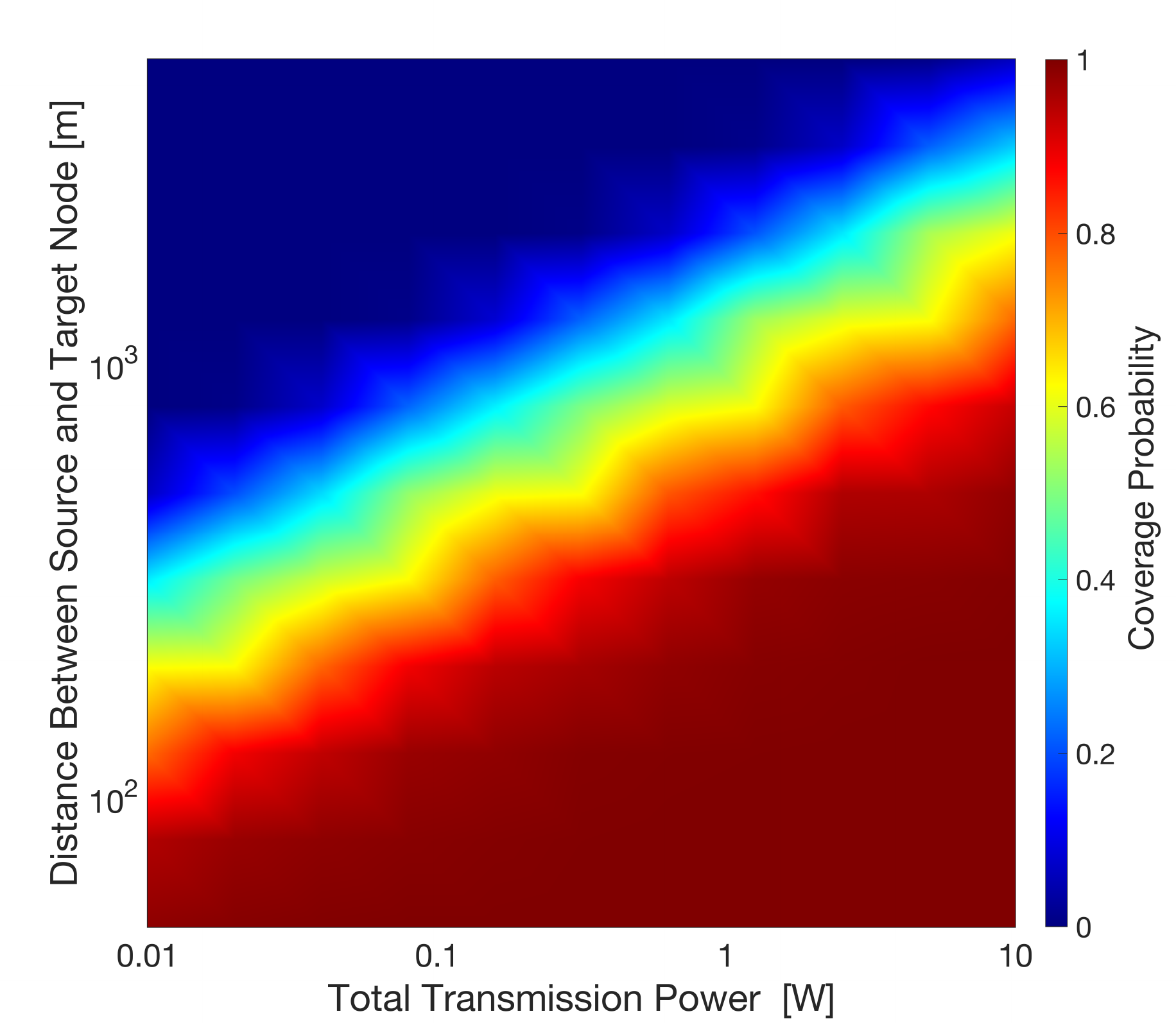}
\caption{Heat Map of Coverage Probability for RF Transmission.}
\label{fig6}
\end{minipage}
\vspace{-6mm}
\end{figure*}

\vspace{-2mm}
\subsection{Comparison between THz and RF Transmission}
In this subsection, we compare the throughput and coverage probability performance between THz and RF communication systems with different transmission power and communication distances.

\par
Fig.~\ref{fig3} illustrates the throughput performance for both THz and RF communications. The overlap between lines and points proves again that the proposed stepwise-optimal routing strategy can approach the ideal upper bound for throughput in most cases. For THz networks, when the power differs by $30$~dB, the throughput differs by tens to hundreds of times. However, for RF networks, a $30$~dB difference in power only results in a throughput difference of several to tens of times. A common characteristic of both networks is that as the distance increases, the gain from increasing power also increases. 
Furthermore, it is evident in Fig.~\ref{fig3} that the throughput of THz networks is more sensitive to changes in distance compared to RF networks. 
When transmitting at high power, THz routing consistently achieves higher throughput compared to RF due to its bandwidth advantage. However, when the transmission power is low, RF networks exhibit higher throughput than THz networks for transmission distances exceeding $50$~m.

In Fig.~\ref{fig4}, under the same condition, RF communication exhibits a much higher coverage probability than THz communication, indicating RF's superior reliability in scenarios requiring robust coverage. 
Even at a coverage threshold of $15$~dB, short-range RF communication can still achieve close to $100\%$ coverage probability. In contrast, even at a coverage threshold of $-15$~dB, the coverage probability of long-range THz communication is close to $0\%$. Furthermore, increasing the transmission power by ten times does not improve the coverage probability as much as reducing the distance to one-tenth of its original value. 

{\color{black} Combining the insights from Fig.~\ref{fig3} and Fig.~\ref{fig4}, we propose a THz and RF relay selection mechanism based on transmission power and communication distance. 
When sufficient power is available and the distance between the source and target nodes is relatively short (less than 500~m), THz is preferable due to its higher throughput enabled by a wide bandwidth. 
However, under power-limited scenarios, THz only outperforms RF within extremely short distances (below 50~m). As the distance increases, RF becomes more advantageous due to its superior signal penetration and lower path loss.}

\vspace{-0.3cm}
\subsection{Transmission Power Design}

{\color{black}
This subsection presents an application of the analytical results, that is, transmission power design. To investigate the suitable range of transmission power, a straightforward approach is examining the coverage probability under different sets of parameters, such as communication distances and transmission power. However, exhaustively evaluating all sets of parameters with traditional simulation methods is computationally expensive. On the contrary, system parameters can be mapped to coverage performance with low complexity using the analytical expressions provided in Theorem~\ref{theorem4} and Theorem~\ref{theorem5}, thus alleviating the computational burden. 
Although certain simplifications are made to enable analytical tractability, the results effectively characterize the statistical behavior of the network and highlight the main trends and trade-offs.

Fig.~6 and Fig.~7 represent heat maps of the coverage probability. 
The x-axis denotes the total transmission power allocated across all hops in the routing path. 
The blue parts in the upper left corner of the figures show that insufficient transmission power can result in communication interruptions. In contrast, the red parts in the lower right corner indicate that excessive power usage can lead to unnecessary energy waste. Therefore, the upper boundary of the dark red area represents the recommended combinations of communication distance and total transmission power, offering an effective balance between energy consumption and reliable communication. 
% Noted that the presented results characterize the statistical behavior of the network, where only the key factors are considered to enable analytical tractability. Despite these idealizations, the heat maps effectively highlight the main trends and trade-offs, providing actionable guidance for system design and parameter selection in real-world applications.
}

\vspace{-2mm}

\section{Insights on Open Issues and Research Directions} \label{section6}

This section explores several factors that align with the analytical framework in this paper. Combining these factors with the proposed framework presents promising future research directions.

\vspace{-2mm}

\subsection{UAV Communication}
{\color{black} UAVs have emerged as a versatile platform in wireless communication due to flexible deployment, attracting increasing research interest. For instance, age of information (AoI)-minimal clustering strategies have been proposed to reduce latency. \cite{abd2018average,hu2020aoi}. Additionally, joint trajectory and transmission optimization in UAV-assisted wireless-powered communication networks (WPCNs) has been studied to improve energy efficiency \cite{oubbati2021multi,liu2024aoi}. SG-based routing enables efficient performance evaluation and low-complexity route planning, making it suitable for energy-constrained UAV networks. 

In the following, we present an extended application, that is, throughput analysis in a UAV-based communication system. 
{\color{black} In simulation, the UAVs are considered to be multi-rotor platforms with hovering capability.}
We consider a ground source node transmitting signals to a ground target node through multi-hop relaying via the UAV network. The UAV-ground link establishes a LoS link with probability 
\begin{equation}
    P_{\rm{LoS}}(\theta_e) = \frac{1}{1 + a\exp \left( { - b\left( \theta_e - a \right)} \right)},
\end{equation}
where $a=25.27$ and $b=0.5$ are terrain-related parameters \cite{al2014optimal}. $\theta_e$ represents the elevation angle between the UAV and the ground source node or target node. The absorption coefficient is set to $0.005$~m$^{-1}$ for LoS and $0.5$~m$^{-1}$ for NLoS. In addition, the terrestrial source node adopts the maximum LoS probability association strategy to find the first relay UAV. The inter-UAV links are assumed to be LoS, thus the intermediate hops of routing (excluding the UAV-ground hops) can follow the proposed routing method and analytical framework in this article. In both Fig.~\ref{fig7} and Fig.~\ref{fig8}, routing is conducted via THz links.

\begin{figure}[tbp]
\vspace{-2mm}
% \begin{minipage}[t]{0.48\linewidth}
\centering
\includegraphics[width=0.7\linewidth]{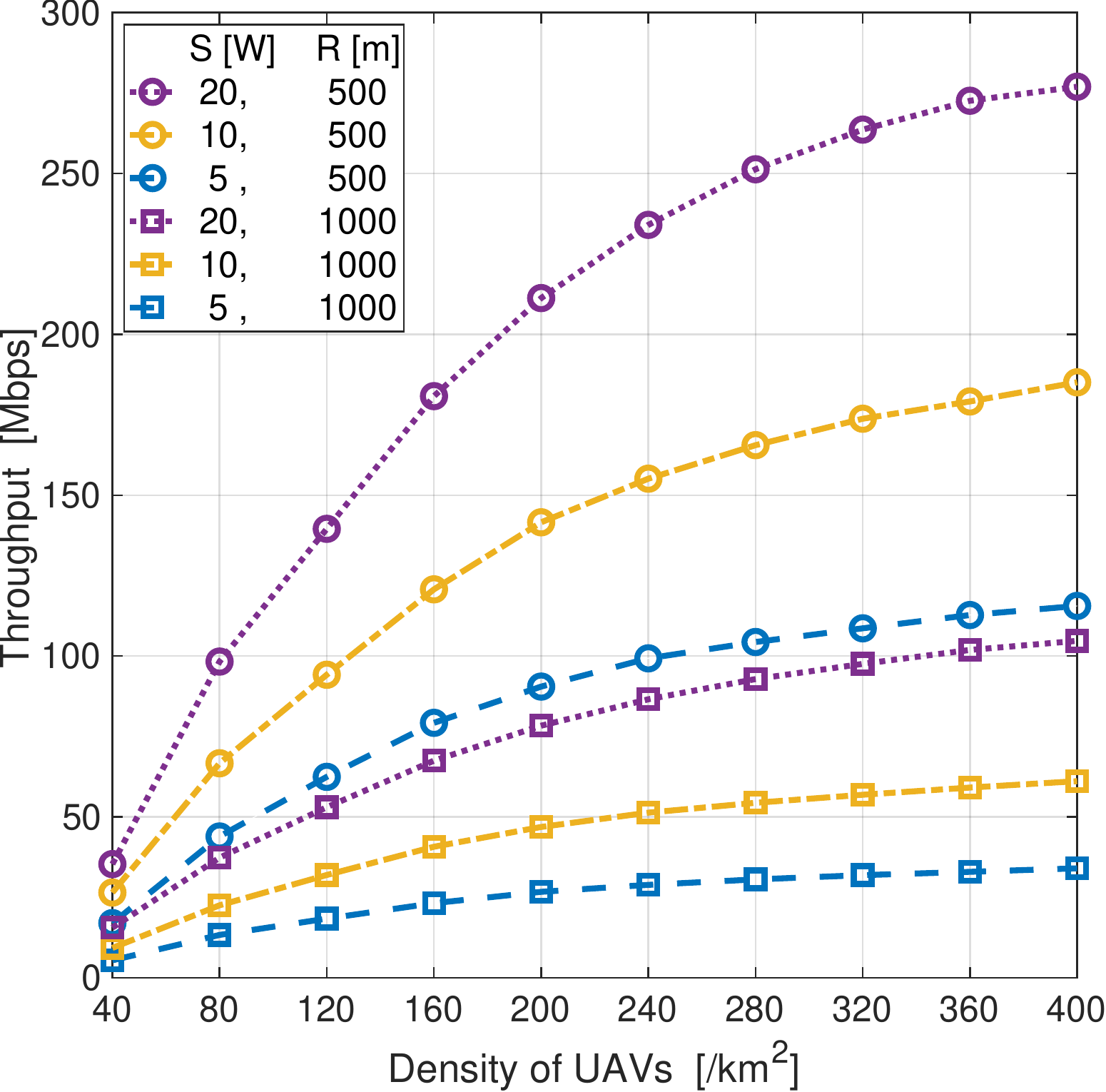}
\caption{Impact of UAV Density on Throughput Performance.}
\label{fig7}
\vspace{-2mm}
\end{figure}
% \end{minipage}
% \hfill
% \begin{minipage}[t]{0.48\linewidth}
\begin{figure}[tbp]
\vspace{-2mm}
\centering
\includegraphics[width=0.7\linewidth]{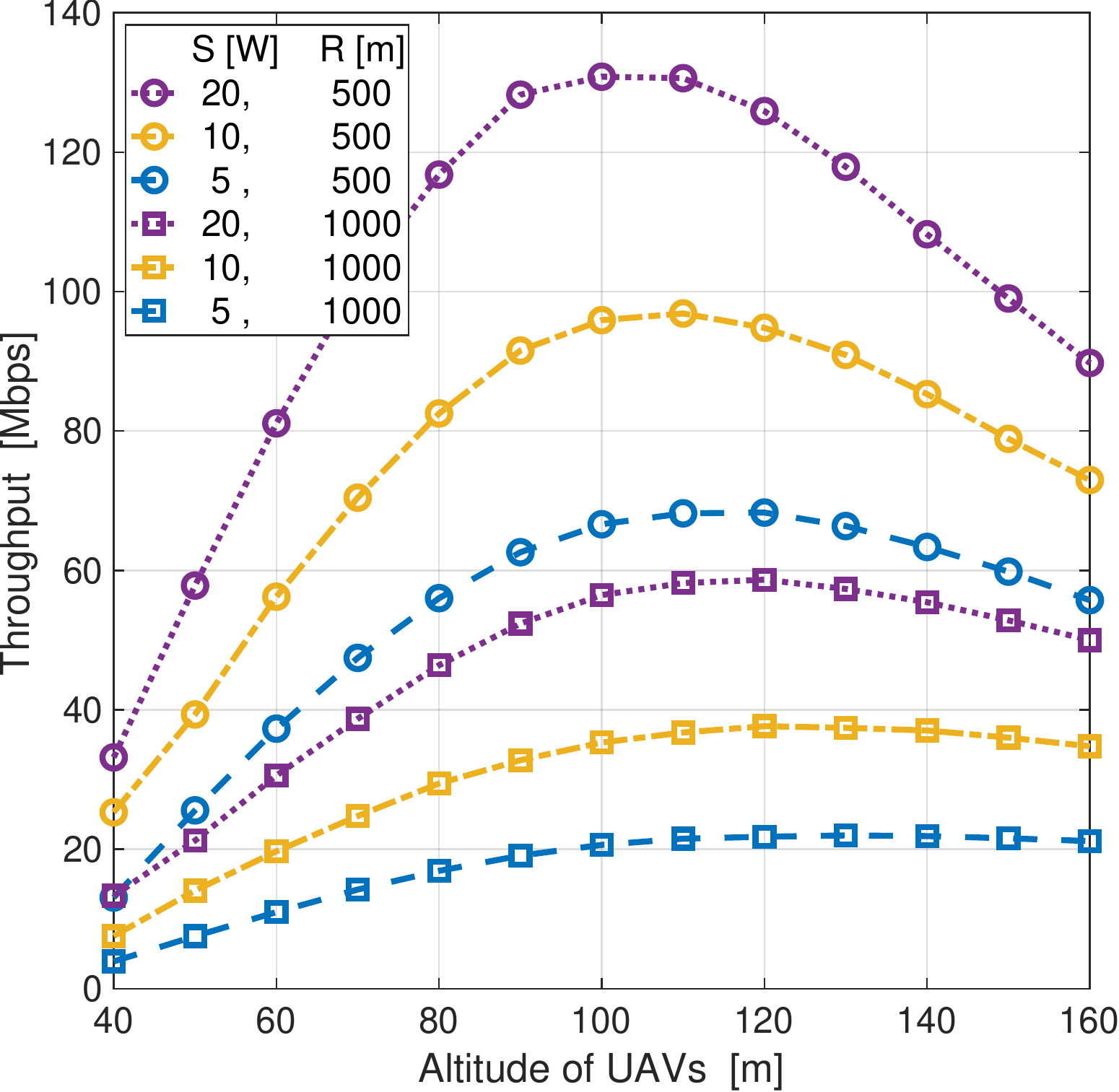}
\caption{Impact of UAV Altitude on Throughput Performance.}
\label{fig8}
\vspace{-4mm}
\end{figure}

As shown in Fig.~\ref{fig7} and Fig.~\ref{fig8}, the throughput decreases to one-third of the original when the communication distance doubles. When the transmission power doubles, the throughput increases to $1.4$~times the original. Fig.~\ref{fig7} shows the impact of adjusting UAV density, with UAVs' altitude fixed at $80$~m. Increasing UAV density allows for the selection of UAV relays closer to the ideal relay position, thereby augmenting throughput. Therefore, throughput increases with UAV density and converges to the performance of ideal routing. In Fig.~\ref{fig8}, we study the influence of UAV altitude when the density of UAVs is fixed at $100$~/km$^2$. 
We observe that the throughput peaks at a UAV altitude of approximately 110~m, which is independent of both the transmission power and the communication distance.
This indicates that when changing the UAV altitude, the throughput is mainly influenced by the UAV-ground link. Increasing the altitude enhances the LoS probability but diminishes the average path loss. The optimal altitude arises from a balance between the LoS probability and path loss.

Combining SG-based routing with UAV communication systems presents several promising research directions. Specifically, UAVs operate in three-dimensional (3D) space with altitude as a critical dimension. How to extend the proposed routing of two-dimensional planes to 3D is worth discussing. Moreover, UAV communication systems exhibit high mobility. Future research can explore analytical frameworks that integrate UAV velocity and trajectory planning into routing decisions.}

%This phenomenon primarily results from the decrease in the average horizontal distance to the associated UAV from ground users and the increase in the elevation angle, consequently enhancing the LoS probability. The trend towards stabilization is attributed to the characteristics of the LoS probability function, where the rate of increase in LoS probability gradually diminishes with an increase in the elevation angle. 

{\color{black}
\subsection{Multi-Source Multi-Destination Communication}

Multi-source multi-destination networks have the potential to significantly enhance spectral efficiency and system throughput through concurrent data flows \cite{ivanescu2022multi,liang2024securing,qu2020reliability}, thereby motivating increasing interest in investigating the advantages of relaying in such architectures. While this work analyzes the single-source single-destination scenario, the proposed SG-based routing framework can be extended to support multi-source multi-destination communications. When sufficient resources are available, each relay node is capable of participating in multiple transmissions simultaneously. Under this condition, the routing strategies developed in this work remain applicable without modification. 

Under resource-constrained conditions, where each node can support only a limited number of simultaneous connections due to hardware limitations, relay nodes involved in flow may become temporarily unavailable for new routes.
This interdependence across sessions introduces contention and must be accounted for in routing decisions. 
To address this, a potential approach is to model relay node availability as a function of traffic load. 
Specifically, the probability that a relay node is available for routing depends on the number of concurrent routes traversing the network. This affects the distribution of available relay nodes, which in turn modifies the single-hop distance distribution used in routing analysis. By deriving corresponding distance distribution, our analytical framework can be extended to multi-source multi-destination networks.}

\section{Conclusion} \label{section7}
\vspace{-1mm}
In this paper, we first derived the maximum throughput strategy for multi-hop routing under ideal scenarios. By making small adjustments to the above strategy, a practical strategy for THz and RF networks called the stepwise-optimal routing strategy was proposed. It has been proven to have significant advantages in throughput compared with the existing long-hop and short-hop strategies. Furthermore, a tradeoff was made for tractability of derivation and a stepwise-suboptimal routing strategy was provided. The stepwise-suboptimal routing strategy has a very close throughput performance to the stepwise-optimal routing strategy. It is suitable for the SG analysis framework and can provide analytical expressions for network performance metrics. Then, we compared the throughput and coverage performance of the above routing strategies in THz and RF networks. Finally, the proposed analytical framework and routing strategies were applied to system parameter design and UAV networks.

% \par
% In future work, the throughput rate will be further increased if the huge data packet is divided into multiple sub-packets transmitted by multiple paths and pipeline protocol \cite{xu2016pipeline}. It is of practical significance to apply the SG-based analysis framework for the above transmission scheme. How to extend the maximum throughput routing of two-dimensional planes to three-dimensional spheres is also worth discussing \cite{wang2022ultra}. Since the inter-satellite link does not suffer from water molecule attenuation, the THz band can bring higher communication quality by its larger bandwidth. The RF frequency band is more suitable for signal transmission in the satellite-ground link. In such a scenario, besides packet buffering latency, signal decoding, and propagation latency are also worth considering \cite{rabjerg2021exploiting}. 

\appendices
\vspace{-2mm}\section{Proof of Proposition~\ref{prop1}} \vspace{-1mm}\label{app:prop1}
We first derive the result for SNR $\gg$ 1 then extend the result for SNR $\ll$ 1. Finally, we prove the uniqueness of the optimal power allocation strategy in both cases.

To obtain the power allocation strategy, we first establish the Lagrange function for the optimization problem $\mathscr{P}_1^Q$, where $Q = \{\rm{THz}, \rm{RF}\}$. When SNR $\gg$ 1,
\begin{IEEEeqnarray}{RCL}
    \mathcal{L} &=& \frac{-1}{\sum_{k=1}^{{K}_{Q}} \frac{1}{\overline{\tau}_{Q,k}\left(r_k,s_k\right)}} + \lambda \left( \sum_{k=1}^{{{K}_{Q}}} s_k - S_Q \right) \notag \\
    &\overset{(a)}{\approx}&  \frac{-B_Q }{\sum_{k=1}^{{K}_{Q}} \frac{1}{ \log_{2}\left(\overline{\rm{SNR}}_{Q,k}\left(r_k,s_k\right) \right)}} + \lambda \left( \sum_{k=1}^{{K}_{Q}} s_k - S_Q \right), \IEEEeqnarraynumspace
\end{IEEEeqnarray}
where $(a)$ is obtained by $\rm{SNR}+1 \approx \rm{SNR}$ for large SNR. In order to simplify the expression, we express $\overline{\rm{SNR}}_{Q,k}\left(r_k,s_k\right)$ as ${\rm{SNR}}_k$ only in Appendix~\ref{app:prop1}. Take the partial derivative with respect to $s_k$, 
\begin{IEEEeqnarray}{RCL}
    \frac{\partial \mathcal{L}}{\partial s_k} &=& \frac{B_Q}{\left( \sum_{k=1}^{{K}_{Q}} \frac{1}{ \log_{2}{\rm{SNR}}_k}\right)^2}  \frac{\partial }{\partial s_k} \left( \sum_{k=1}^{{K}_{Q}} \frac{1}{ \log_{2}{\rm{SNR}}_k}\right) \!+\! \lambda  \notag \\
    &=& \frac{B_Q}{\left( \sum_{k=1}^{{K}_{Q}} \frac{1}{ \log_{2}{\rm{SNR}}_k}\right)^2} \frac{\partial }{\partial s_k} \left( \frac{1}{ \log_{2}{\rm{SNR}}_k}\right) + \lambda. \label{appenA_2} \IEEEeqnarraynumspace
\end{IEEEeqnarray}
The partial derivative term in (\ref{appenA_2}) can be further simplified as follows,
\begin{IEEEeqnarray}{RCL}\label{appenA_3}
    \frac{\partial }{\partial s_k}  \frac{1}{ \log_{2}{\rm{SNR}}_k} &=&  \frac{-1}{\left( \log_{2}{\rm{SNR}}_k \right)^2}  \frac{ 1}{{\rm{SNR}}_k \ln2} \frac{\partial {\rm{SNR}}_k}{\partial s_k}  \notag \IEEEeqnarraynumspace
    \\
    &\overset{(b)}{=}& - \frac{1}{\left( \log_{2}{\rm{SNR}}_k \right)^2 s_k \ln2 },
\end{IEEEeqnarray}
where $(b)$ follows the equation that $\frac{\partial {\rm{SNR}}_k}{\partial s_k} = \frac{{\rm{SNR}}_k}{s_k}$. Substitute (\ref{appenA_3}) into (\ref{appenA_2}), and set $\frac{\partial \mathcal{L}}{\partial s_k}=0$, we get
\begin{IEEEeqnarray}{RCL}\label{appenA_4}
    \lambda &=& \frac{B_Q}{\left( \sum_{k=1}^{{K}_{Q}} \frac{1}{ \log_{2}{\rm{SNR}}_k}\right)^2} \frac{1}{\left( \log_{2}{\rm{SNR}}_k \right)^2 s_k \ln2 }, \ \forall k \IEEEeqnarraynumspace
\\ \label{appenA_5}
    &=& \frac{B_Q}{\left( \sum_{j=1}^{{K}_{Q}} \frac{1}{ \log_{2}{\rm{SNR}}_j}\right)^2}  \frac{1}{\left( \log_{2}{\rm{SNR}}_j \right)^2 s_j \ln2}, \ \forall j. \IEEEeqnarraynumspace
\end{IEEEeqnarray}
Consider that the first term in (\ref{appenA_4}) or (\ref{appenA_5}) does not change with $k$ or $j$, the result for large SNR can be given as,
\begin{equation}\label{appenA_6}
    \log_{2}{\rm{SNR}}_k \sqrt{s_k} = \log_{2}{\rm{SNR}}_j \sqrt{s_j}, \ \forall j,k.
\end{equation}

When SNR $\ll$ 1, $\log_2\left(1 + \rm{SNR} \right) \approx \rm{SNR}/\ln{2}$ holds, and
\begin{IEEEeqnarray}{RCL}
    \frac{\partial \mathcal{L}}{\partial s_k} & \approx& \frac{\partial }{\partial s_k} \left( - \frac{B_Q }{\sum_{k=1}^{{K}_{Q}} \frac{\ln{2}}{ { \rm{SNR}}_k }} + \lambda \left( \sum_{k=1}^{{{K}_{Q}}} s_k - S_Q \right) \right)  \notag\\
    & =& \frac{-B_Q \ln{2}}{\left(\sum_{k=1}^{{K}_{Q}} \frac{\ln{2}}{ { \rm{SNR}}_k } \right)^2  {\rm{SNR}}_k  s_k} + \lambda.
\end{IEEEeqnarray}
Taking the partial derivative and setting it equal to zero, we obtain
\begin{equation}
    \lambda = \frac{B_Q \ln{2}}{\left(\sum_{k=1}^{{K}_{Q}} \frac{\ln{2}}{ { \rm{SNR}}_k } \right)^2 \cdot {\rm{SNR}}_k \cdot s_k}, \ \forall k.
\end{equation}
Similar to the derivation of SNR $\gg$ 1, the result for SNR $\ll$ 1 can be obtained,
\begin{equation}
    {\rm{SNR}}_k s_k = {\rm{SNR}}_j s_j, \ \forall j,k.
\end{equation}

Finally, given a fixed $r_k$, both $\log_{2}{\rm{SNR}}_k \sqrt{s_k}$ and ${\rm{SNR}}_k s_k$ increase monotonically with $s_k$. When SNR $\gg$ 1, given that $s_1^*,s_2^*,\dots,s_{{K}_{Q}}^*$ is a set of the optimal power allocation strategy. The relationship in (\ref{appenA_6}) shows that when $s_k^*$ increases, all the other $s_j^*,\,j=1,2,\dots,{{K}_{Q}}$ will increase, resulting in the constraint (\ref{st:constraint1-1}) no longer being satisfied, and vice versa. Therefore, $s_1^*,s_2^*,\dots,s_{{K}_{Q}}^*$ is the unique the optimal power allocation strategy. The proof of uniqueness when SNR $\ll$ 1 is similar to that of SNR $\gg$ 1, therefore, is omitted here.

\vspace{-2mm}\section{Proof of Proposition~\ref{prop2}}\label{app:prop2}

When SNR $\gg$ 1, we define $\mathcal{C} =\sqrt{s_k^*} \log_2 \overline{\rm{SNR}}_k\left(r_k,s_k^*\right)$ and Proposition~\ref{prop1} demonstrates that $\mathcal{C}$ is a constant independent of $k$. Thus, the optimization objective can be approximated as follows:
\begin{IEEEeqnarray}{RCL}
    \overline{\tau}_Q^T &\approx& \frac{B_Q}{\sum_{k=1}^{{K}_{Q}} \frac{1}{\log_2 \left(\overline{\rm{SNR}}_k\left(r_k,s_k^*\right)\right)} } \notag\\
    &=& \frac{B_Q \, \mathcal{C}}{\sum_{k=1}^{{K}_{Q}} \sqrt{s_k^*}}.
\end{IEEEeqnarray}
Although the relationship between $\overline{\tau}_Q^T$ and $r_k$ is not explicit, both $s^*$ and $\mathcal{C}$ show symmetry for different $r_k$. Therefore, $\frac{\partial \overline{\tau}_Q^T}{\partial r_k}=\frac{\partial \overline{\tau}_Q^T}{\partial r_j}, \ \forall j,k$ and further we have $r_k^*=r_j^*,\ \forall j,k$. 
\par
Note that in Proposition~\ref{prop1}, when the given values of $r_1,r_2,\dots,r_{{K}_{Q}}$ are different, the values of $s_1^*,s_2^*,\dots,s_{{K}_{Q}}^*$ are unique and different. Interchanging the values of $s_k^*$ and $s_j^*$ will result in a not optimal power allocation strategy. In Proposition~\ref{prop2}, the power allocation strategy is adopted, so any change in $r_k$ leads to a corresponding change in $s_k^*$. Therefore, the symmetry property is applicable only in Proposition~\ref{prop2} and not in Proposition~\ref{prop1}.

When SNR $\ll$ 1, we define $\mathcal{C} =s_k^* \cdot \overline{\rm{SNR}}_k\left(r_k,s_k^*\right)$. The optimization objective can be approximated as follows
\begin{IEEEeqnarray}{RCL}
    \overline{\tau}_Q^T &\approx& \frac{B_Q}{\sum_{k=1}^{{K}_{Q}} \frac{\ln{2}}{\overline{\rm{SNR}}_k\left(r_k,s_k^*\right)} } \notag\\
    &=& \frac{B_Q \, \mathcal{C}}{ S_Q \ln{2}},
\end{IEEEeqnarray}
where $S_Q$ is the total transmission power. When power allocation strategy is employed, both $\overline{\tau}_{\rm{THz}}^T$ and $\overline{\tau}_{\rm{RF}}^T$ become constants that are independent of $r_k$.

\vspace{-2mm}\section{Proof of Proposition~\ref{prop3}}\label{app:prop3}
In what follows, we discuss the choices of the optimal number of hops in four cases. \\
\textit{Case 1:} For THz network, when SNR $\ll$ 1, the average throughput can be approximated as,
\begin{IEEEeqnarray}{RCL}  
\overline{\tau}_{\rm{THz}}^T & \approx&  \frac{B_{\rm{THz}} \overline{\rm{SNR}}_{{\rm{THz}},k}\left(\frac{R}{{K}_{\rm{THz}}},\frac{S}{{K}_{\rm{THz}}}\right)}{\ln{2}\cdot {{K}_{\rm{THz}}}}   \notag\\
& =&  \frac{B_{{\rm{THz}}}S \,G_{{\rm{THz}}} \, \eta _{{\rm{THz}}}}{ \ln{2}\cdot \sigma_{\rm{THz}}^2R^2}\exp\left({\frac{-\beta_{\rm{THz}} R}{{{K}_{\rm{THz}}}}}\right).
\end{IEEEeqnarray}
Therefore, we can conclude that as ${{K}_{\rm{THz}}}$ increases, the average throughput $\overline{\tau}_{\rm{THz}}^T$ increases.

\par

\textit{Case 2:} For THz network, when SNR $\gg$ 1,
\begin{IEEEeqnarray}{RCL} 
    &\overline{\tau}_{\rm{THz}}^T  \approx& \frac{B_{\rm{THz}} }{K_{\rm{THz}}}\log _2 \left( \overline{\rm{SNR}}_{\rm{THz},k}\left(\frac{R}{K_{\rm{THz}}},\frac{S}{K_{\rm{THz}}}\right) \right) \notag
    \\ & = \frac{B_{\rm{THz}} }{K_{\rm{THz}}}& \!\left(\! \log _2 \frac{S G_{{\rm{THz}}}  \eta _{{\rm{THz}}}}{\sigma_{\rm{THz}}^2R^2} \!-\! \frac{\beta_{\rm{THz}} R}{K_{\rm{THz}}\ln 2}\!+\!\log _2K_{\rm{THz}}\right ). \IEEEeqnarraynumspace 
\end{IEEEeqnarray}  
To facilitate the process of differentiation, we introduce a variable $x=\frac{1}{K_{\rm{THz}}}$. 
% As a result
% \begin{equation}    
%  \overline{\tau}_{\rm{THz}}^T \! \approx \! {B_{\rm{THz}} }{x} \! \left(\! \log _2 \! \frac{S G_{{\rm{THz}}}  \eta _{{\rm{THz}}}}{\sigma_{\rm{THz}}^2R^2} \!-\! \frac{\beta_{\rm{THz}} Rx}{\ln 2}\!-\!\log _2x \! \right ). \IEEEeqnarraynumspace 
% \end{equation}
Next, we calculate the partial derivative of $\overline{\tau}_{\rm{THz}}$ with respect to $K_{\rm{THz}}$ and equate it to zero
\begin{IEEEeqnarray}{RCL} 
    \frac{\partial \overline{\tau}_{\rm{THz}}^T }{\partial K_{\rm{THz}}}  &=& \frac{\partial \overline{\tau}_{\rm{THz}}^T }{\partial x}\frac{\partial x }{\partial K_{\rm{THz}}}  \notag \\
   & \approx& \frac{B_{\rm{THz}}}{K_{\rm{THz}}^3 \ln2} \Big( K_{\rm{THz}}+2R\beta_{{\rm{THz}}} \notag \\ & &-  K_{\rm{THz}}\ln K_{\rm{THz}} 
   - K_{\rm{THz}}\ln \frac{S \,G_{{\rm{THz}}} \, \eta _{{\rm{THz}}}}{\sigma_{\rm{THz}}^2R^2}\Big),
\IEEEeqnarraynumspace 
\end{IEEEeqnarray} 
\begin{equation}
    \frac{2R\beta_{{\rm{THz}}}}{K_{\rm{THz}}} - \ln \frac{S \,G_{{\rm{THz}}} \, \eta _{{\rm{THz}}}}{\sigma_{\rm{THz}}^2R^2} + 1 = \ln K_{\rm{THz}}.
\end{equation}
$K_{\rm{THz}}$ can be obtained from the above transcendental equation. From the monotonically decreasing property of inverse proportional function and monotonically increasing property of logarithmic function, $K_{\rm{THz}}$ has a unique solution.

\textit{Case 3:} For RF network, when SNR $\ll$ 1, the average throughput can be approximated as,
\begin{IEEEeqnarray}{RCL}
        \overline{\tau}_{\rm{RF}}^T  & \approx& \frac{B_{\rm{RF}}\overline{\rm{SNR}}_{{\rm{RF}},k}\left(r_k,s_k\right)}{\ln{2} \cdot K_{\rm{RF}}} \notag \\
        &=& \frac{B_{\rm{RF}}}{\ln{2} \cdot K_{\rm{RF}}} \frac{S}{K_{\rm{RF}}} \frac{G_{{\rm{RF}}} \eta _{{\rm{RF}}} }{\sigma_{\rm{RF}}^2}\left ( \frac{K_{\rm{RF}}}{R} \right )^{\beta_{\rm{RF}}} \notag\\
        &=& \frac{B_{\rm{RF}} S G_{{\rm{RF}}} \eta _{{\rm{RF}}}K_{\rm{RF}}^{\beta_{\rm{RF}}-2}}{\ln{2} \cdot R^{\beta_{\rm{RF}}}\sigma_{\rm{RF}}^2},
\end{IEEEeqnarray}
This shows that when $\beta_{\rm{RF}} \geq 2$, increasing $K_{\rm{RF}}$ leads to an increase in the average throughput $\overline{\tau}_{\rm{RF}}^T$. It is worth noting that in numerous studies, RF devices are typically equipped with omnidirectional antennas. Furthermore, it is often assumed that $\beta_{\rm{RF}} = 2$ in an ideal channel environment, as achieving a path loss exponent lower than 2 can be challenging.

\textit{Case 4:} For THz network, when SNR $\gg$ 1, the average throughput can be approximated as
\begin{IEEEeqnarray}{RCL}
     &\overline{\tau}_{\rm{RF}}^T  \approx& \frac{B_{\rm{RF}} }{K_{\rm{RF}}}\log _2\overline{\rm{SNR}}_{\rm{RF},k}\left(r_k,s_k\right) \notag\\ 
     &= \frac{B_{\rm{RF}} }{K_{\rm{RF}}}& \log _2\frac{S G_{{\rm{RF}}}  \eta _{{\rm{RF}}} \left ( \frac{K_{\rm{RF}}}{R} \right )^{\beta_{\rm{RF}}} }{K_{\rm{RF}}\sigma_{\rm{RF}}^2}\notag\\ 
     &=\frac{B_{\rm{RF}} }{K_{\rm{RF}}}& \left ( \log _2\frac{SG_{{\rm{RF}}} \eta _{{\rm{RF}}}}{\sigma_{\rm{RF}}^2} \!-\! \log _2 K_{\rm{RF}}\!-\!\beta_{\rm{RF}}\log_2\frac{R}{K_{\rm{RF}}}\right ). \IEEEeqnarraynumspace 
\end{IEEEeqnarray}

By taking the partial derivative of $\overline{\tau}_{\rm{THz}}^T$ with respect to $K_{\rm{RF}}$ and equating it to zero, we obtain
\begin{IEEEeqnarray}{RCL}
     \frac{\partial \overline{\tau}_{\rm{RF}}^T }{\partial K_{\rm{RF}}}  &=& \frac{B_{\rm{RF}}}{K_{\rm{RF}}^2 \ln2}\Big(-1+\beta_{\rm{RF}}+\ln K_{\rm{RF}} \notag \\
    & &+\beta_{\rm{RF}}\ln\frac{R}{K_{\rm{RF}}}+\ln \frac{S \,G_{{\rm{RF}}} \, \eta _{{\rm{RF}}}}{\sigma_{\rm{RF}}^2}\Big),
\IEEEeqnarraynumspace 
\end{IEEEeqnarray}

\begin{equation}
    K_{\rm{RF}}\left ( \frac{R}{K_{\rm{RF}}} \right )^{\beta_{\rm{RF}}} \!= \!\exp\left ( 1\!-\! {\beta_{\rm{RF}}}\!+\! \ln \frac{S \,G_{{\rm{RF}}}  \eta _{{\rm{RF}}}}{\sigma_{\rm{RF}}^2}\right ).
\end{equation}
Due to the assumption $\beta_{\rm{RF}} \geq 2$, the equation for $K_{\rm{RF}}$ has a unique solution.

\vspace{-2mm}\section{Proof of Lemma~\ref{lemma1}}\label{app:lemma1}
We start with the \ac{CDF} of type \uppercase\expandafter{\romannumeral1} distance distribution, and by taking the derivative of the \ac{CDF}, we can obtain the corresponding \ac{PDF}.

\begin{figure}[ht]
\vspace{-0.2cm}
	\centering
	\includegraphics[width=0.8\linewidth]{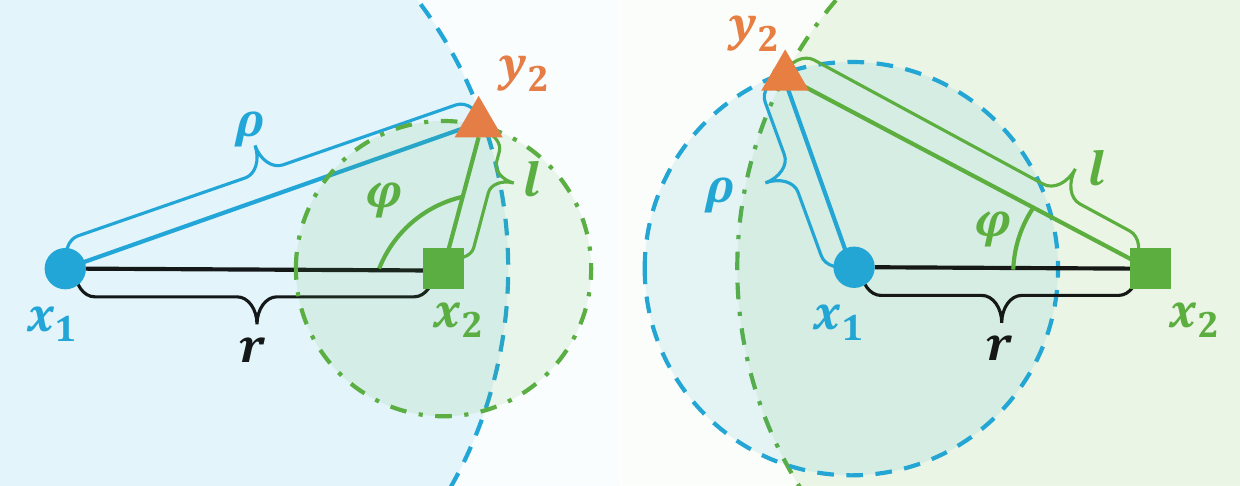}
	\caption{Explanatory figure of the type \uppercase\expandafter{\romannumeral1} distance distribution ($\rho > r$ in the left, and $\rho < r$ in the right).}
	\label{fig:figure1}
	\vspace{-0.2cm}
\end{figure}

As shown in Fig.~\ref{fig:figure1}, the center of the blue dash circle represents the source node or target node and is denoted as $x_1$. The center of the green dash-dot circle corresponds to the first or last ideal relay position in this case, and is denoted as $x_2$. The orange triangle $y_2$ that lies on the boundaries of both circles represents the selected relay node that is closest to $x_2$ among all the points from the PPP. The distance between $x_1$ and $x_2$ is denoted as $r$, which is a constant. The distances from $y_2$ to $x_1$ and $x_2$ are denoted as $\rho$ and $l$, respectively.

\par
From the definition, the CDF of type \uppercase\expandafter{\romannumeral1} distance distribution represents the probability that the selected relay $y_2$ is inside the blue dash circle. Therefore, the CDF can be obtained by integrating over the blue dash circle,
\begin{IEEEeqnarray}{RCL}\label{D-1}
    F_{Q}^{(1)} (\rho | r) = \int_{r-\rho}^{r+\rho} \frac{2 \varphi \left(\rho, r,l \right)}{2 \pi}  f_{Q}^l \left ( l \right ) \mathrm{d} l,
\end{IEEEeqnarray}
where $f_{Q}^l \left ( l \right )$ is the PDF of $l$, representing the probability that the distance between the selected point $y_2$ and $x_2$ is equal to $l$. Since $f_{Q}^l \left ( l \right )$ corresponds to the entire green dash-dot circle, and only a part of it is inside the blue dash circle, a weighting factor of ${2\varphi \left(\rho, r,l \right)}/{2\pi}$ is applied. Using the law of cosines for triangles, 
\begin{equation}\label{D-2}
    \varphi \left(\rho, r,l \right) = \arccos \left ( \frac{r^2+l^2-{\rho}^2}{2rl} \right ).
\end{equation}
\par
The \ac{CCDF} of $l$, denoted as $\overline{F}_{Q}^l \left ( l \right )$, is known as the null probability of PPP,
\begin{IEEEeqnarray}{RCL} \label{D-3}
    \overline{F}_{Q}^l \left ( l \right ) &=& \mathbb{P}\left [ {\rm{No\ relay\ closer\ than \ }}l \right ] \notag \\
    &=& \exp\left ( - \lambda _{Q} \pi l^2 \right ). \IEEEeqnarraynumspace 
\end{IEEEeqnarray}
We derive the PDF of $l$ by taking the derivative of CCDF, 
\begin{IEEEeqnarray}{RCL}\label{D-4}
        f_{Q}^l \left ( l \right ) &=&-\frac{{\mathrm{d}} \overline{F}_{Q}^l \left(l\right)}{{\mathrm{d}} l} \notag \\
        &=& 2 \pi l\lambda _{Q}  \exp\left ( - \lambda _{Q} \pi l^2 \right ).
\end{IEEEeqnarray}
Substituting (\ref{D-2}) and (\ref{D-4}) into (\ref{D-1}), the CDF of type \uppercase\expandafter{\romannumeral1} distance distribution is expressed as follows,
\begin{IEEEeqnarray}{RCL}\label{D-5}
    F_{Q}^{(1)} (\rho | r) &=& \int_{r-\rho}^{r+\rho} 2 \arccos \left ( \frac{r^2+l^2-{\rho}^2}{2rl} \right )  \notag\\
    & &\times  \lambda_Q \exp\left ( - \lambda _{Q} \pi l^2 \right ) |l| \mathrm{d} l.
\end{IEEEeqnarray}
Note that the inequality $r-\rho<0$ may hold, and $l$ might be negative. Therefore, we discuss the derivation in two conditions according to whether $\rho$ is greater than $r$, which are shown in the left and right parts of Fig.~\ref{fig:figure1} respectively. The final expressions are almost consistent in both conditions, with the addition of an absolute value sign to the last term in the integral for $l$.

\par
Finally, the PDF of type \uppercase\expandafter{\romannumeral1} distance distribution is obtained by differentiating the CDF using Leibniz's rule,
\begin{IEEEeqnarray}{RCL}  \label{D-6}
    f_{Q}^{(1)} (\rho | r) &=& \frac{{\mathrm{d}} F_{Q}^{(1)} (\rho | r) }{{\mathrm{d}} \rho} \notag \\
    &=& \int_{r-\rho}^{r+\rho} \frac{2 \lambda_Q \rho \exp \left( - \lambda_Q \pi l^2 \right) } {r \sqrt{ 1 - \frac{\left(r^2 + l^2 - \rho^2  \right)^2 }{4 r^2 l^2} }} \mathrm{d} l \notag \\
    & &+ 2 \arccos \left ( \frac{r^2+(r+\rho)^2-{\rho}^2}{2r(r+\rho)} \right )  \notag \\ 
    & & \times \lambda_Q \exp\left ( - \lambda _{Q} \pi (r+\rho)^2 \right ) (r+\rho) \notag \\
    & &- 2 \arccos \left ( \frac{r^2+(r-\rho)^2-{\rho}^2}{2r(r-\rho)} \right )  \notag \\ & & \times  \lambda_Q \exp\left ( - \lambda _{Q} \pi (r-\rho)^2 \right ) |r-\rho|  \notag  \\
    &=&\int_{r-\rho}^{r+\rho} \frac{2 \lambda_Q \rho \exp \left( - \lambda_Q \pi l^2 \right)} {r \sqrt{ 1 - \frac{\left(r^2 + l^2 - \rho^2  \right)^2 }{4 r^2 l^2} }} \mathrm{d} l.
\end{IEEEeqnarray}
Note that $\arccos \left ( \frac{r^2+(r+\rho)^2-{\rho}^2}{2r(r+\rho)} \right )$$ = \arccos \left ( \frac{r^2+(r-\rho)^2-{\rho}^2}{2r(r-\rho)} \right ) $$= 0$. In addition, the improper integral will appear when $l=0$, and the integrand function is still finite after taking the limit, so the final result will not be affected.

\vspace{-2mm}\section{Proof of Lemma~\ref{lemma2}}\label{app:lemma2}
In the following, we derive $f_{Q}^{(2)} (\rho | r)$ from $f_{Q}^{(1)} (\rho | r)$. As shown in Fig.~\ref{fig:figure2}, the purple diamond $y_1$ is the selected relay that is closest to $x_1$ among all the points in the PPP. In this context, $x_1$, $x_2$ and $y_2$ have the same definitions as provided in Appendix~\ref{app:lemma1}. 
\par

\begin{figure}[ht]
\vspace{-0.2cm}
	\centering
	\includegraphics[width=0.8\linewidth]{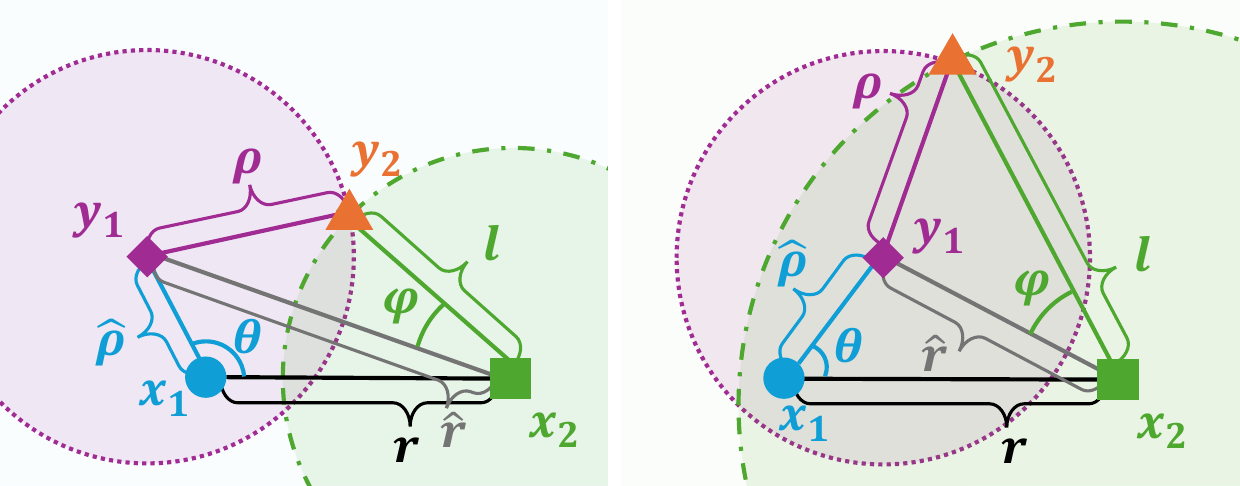}
	\caption{Explanatory figure of the type \uppercase\expandafter{\romannumeral2} distance distribution ($\rho > \widehat{r}$ in the left, and $\rho < \widehat{r}$ in the right).}
	\label{fig:figure2}
	\vspace{-0.2cm}
\end{figure}

In Appendix~\ref{app:lemma1}, the derivation of the type \uppercase\expandafter{\romannumeral1} distance distribution, which represents the distribution of the distance between $x_1$ and $y_2$, has been provided. Due to the homogeneity of the PPP, the symmetry of the distribution indicates that the distance between $y_1$ and $x_2$ also follows type \uppercase\expandafter{\romannumeral1} distance distribution. 
Furthermore, $x_1$ can be regarded as being located at the same position as $y_1$ in type \uppercase\expandafter{\romannumeral1} distance distribution. However, the distance between $x_1$ and $y_1$ forms the distribution given in (\ref{D-4}) for type \uppercase\expandafter{\romannumeral2} distance distribution. The probability of $y_1$ is $r$ away from $x_1$ is ${f_Q^l (\widehat{\rho})}/{2 \pi \widehat{\rho}}$. Therefore, the type \uppercase\expandafter{\romannumeral2} distance distribution $F_{Q}^{(2)}(\rho|r)$ can be derived by traversing the whole plane with $ \widehat{\rho}$ as the integrate variable,
\begin{equation}\label{E-1}
\begin{split}
   F_{Q}^{(2)}(\rho|r) =\int_{0}^{\infty }\int_{0}^{2\pi} \frac{f_Q^l (\widehat{\rho})}{2 \pi \widehat{\rho}} F_Q^{(1)}\left( \widehat{r}\left(  \widehat{\rho},\theta \right) | r \right) \mathrm{d} \theta \mathrm{d} \widehat{\rho}, 
\end{split}
\end{equation}
where $\widehat{r}\left( \widehat{\rho},\theta \right)$ is obtained by the law of cosines for triangles, 
\begin{equation}\label{E-2}
    \widehat{r}\left(\widehat{\rho},\theta \right) =\sqrt{r^2+\hat{\rho}^2-2r \hat{\rho} \cos\theta}.
\end{equation}
Substituting (\ref{D-4}), (\ref{D-5}) and (\ref{E-2}) into (\ref{E-1}), the type \uppercase\expandafter{\romannumeral2} distance distribution $F_{Q}^{(2)}(\rho|r)$ is given,
\begin{IEEEeqnarray}{RCL}\label{E-3}
   F_{Q}^{(2)}(\rho|r) &=& \int_{0}^{\infty }\int_{0}^{2\pi}\lambda _{Q}  \exp\left ( - \lambda _{Q} \pi {\hat{\rho}}^2 \right ) \notag \\ 
   & & \times \int_{\widehat{r}-{\rho}}^{\widehat{r}+{\rho}}2\arccos \left ( \frac{\widehat{r}^2+l^2-{\rho}^2}{2\widehat{r}l} \right ) \notag  \\ 
   & &\times l \cdot  \lambda_Q \exp\left ( - \lambda _{Q} \pi l^2 \right ) \mathrm{d} l \mathrm{d} \theta \mathrm{d} \hat{\rho}. 
\end{IEEEeqnarray}
Finally, the PDF of type \uppercase\expandafter{\romannumeral2} distance distribution is obtained by differentiating the CDF using Leibniz's rule,
\begin{IEEEeqnarray}{RCL}\label{E-4}
    f_{Q}^{(2)} (\rho | r) &=& \frac{{\mathrm{d}} F_{Q}^{(2)} (\rho | r) }{{\mathrm{d}} \rho} \notag \\ 
    &=& \int_{0}^{\infty }\int_{0}^{2\pi}\lambda _{Q}  \exp\left ( - \lambda _{Q} \pi {\hat{\rho}}^2 \right ) \notag \\ 
    && \times \int_{\widehat{r}-\rho}^{\widehat{r}+\rho} \frac{2 \lambda_Q \rho \exp \left( - \lambda_Q \pi l^2 \right) } {\widehat{r} \sqrt{ 1 - \frac{\left(\widehat{r}^2 + l^2 - \rho^2  \right)^2 }{4 \widehat{r}^2 l^2} }} \mathrm{d} l \mathrm{d} \theta \mathrm{d} \hat{\rho}. \IEEEeqnarraynumspace
\end{IEEEeqnarray}

\vspace{-2mm}\section{Proof of Theorem~\ref{theorem2}}\label{app:theorem2}
In the case where each hop has the same transmission power and experiences the same small-scale fading, the only factor that leads to different throughput values is the variation in distance distributions. Based on Lemma~\ref{lemma1} and Lemma~\ref{lemma2}, it is observed that $\tau_{{\rm{THz}},1}=\tau_{{\rm{THz}},K_{\rm{THz}}}$, and $\tau_{{\rm{THz}},2}=\tau_{{\rm{THz}},3}=\dots=\tau_{{\rm{THz}},K_{\rm{THz}}-1}$ for any other $k \neq 1,K_{\rm{THz}}$. From the definition given in (\ref{total tau}), the total throughput of THz networks is,
\begin{IEEEeqnarray}{RCL} \label{F-1}
    \tau_{\rm{THz}}^T &=& \frac{1}{\frac{2}{\tau_{{\rm{THz}},1}}+\frac{K_{\rm{THz}}-2}{\tau_{{\rm{THz}},2}}} \notag \\
     &=& \frac{{\tau_{{\rm{THz}},1}}\cdot {\tau_{{\rm{THz}},2}}}{2{\tau_{{\rm{THz}},2}}+\left( K_{\rm{THz}}-2 \right){\tau_{{\rm{THz}},1}}}.
\end{IEEEeqnarray}
\par
For a non-negative variable $X$, we can calculate the expectation of it by $\mathbb{E}[X] = \int_{0}^{\infty }\overline{F}(x)\mathrm{d} x$, where $ \overline{F}(x)$ is the \ac{CCDF} of $X$.
Therefore, $\tau_{{\rm{THz}},k}$, can be calculated by,
\begin{IEEEeqnarray}{RCL}\label{F-2} 
    \tau_{{\rm{THz}},k} & =& \int_{0}^{\infty }\overline{F}_{\tau_{{\rm{THz}},k}}\left(t\right)\mathrm{d} t \notag \\
     & =& \int_{0}^{\infty} \mathbb{P} \left[ B_Q \log_2 (1+{\rm{SNR}}_{{\rm{THz}},k}) > t \right ]\mathrm{d} t \notag \\
    & =& \int_{0}^{\infty }  \int_{0}^{\infty } f_{{\rm{THz}}}^{(k)}\left(\rho \bigg| \frac{R}{K_{\rm{THz}}} \right) \notag \\ 
    && \times \mathbb{P} \left [ {\rm{SNR}}_{{\rm{THz}},k} > 2^{\frac{t}{{B_{\rm{THz}}}}}-1  \Big| r = l\right ] \mathrm{d} \rho \mathrm{d} t,
    \end{IEEEeqnarray}
where $k=1,2$, $f_{{\rm{THz}}}^{(1)}\left(\rho \big| \frac{R}{K_{\rm{THz}}} \right)$, $f_{{\rm{THz}}}^{(2)}\left(\rho \big| \frac{R}{K_{\rm{THz}}} \right)$ are defined in (\ref{PDF_1}) and (\ref{PDF_2}). Furthermore, the conditional probability in (\ref{F-2}) is calculated by,
\begin{IEEEeqnarray}{RCL}\label{F-3}
     &\!\!\mathbb{P} & \left [ {\rm{SNR}}_{{\rm{THz}},k} > 2^{\frac{t}{{B_{\rm{THz}}}}}-1  \Big| r = \rho \right ] \notag \\ 
    & =& \mathbb{P}  \left [ \mathcal{X}_{\rm{THz}} >   \frac{ \left(2^{\frac{t}{{B_{\rm{THz}}}}}-1\right) \rho^2 \exp\left(\beta_{\rm{THz}} \rho \right) \sigma_{\rm{THz}}^2 }  {\frac{S}{K_{\rm{THz}}} G_{{\rm{THz}}} \eta_{{\rm{THz}}}}  \right ] \notag \\
    &\overset{(a)}{=}& \frac{\Gamma \left(\mu, \mu \, \left( \left(2^{\frac{t}{{B_{\rm{THz}}}}}-1\right) \frac{ \rho^2 \exp\left(\beta_{\rm{THz}} \rho \right) \sigma_{\rm{THz}}^2 }  {\frac{S}{K_{\rm{THz}}} G_{{\rm{THz}}} \eta_{{\rm{THz}}}} \right)^{\frac{\alpha}{2}} \right)}{\Gamma\left(\mu\right)}, \IEEEeqnarraynumspace 
    \end{IEEEeqnarray}
where (a) follows the definition of the \ac{CCDF} of $\mathcal{X}_{\rm{THz}}$. Substitute (\ref{F-2}) into (\ref{F-3}), 
\begin{IEEEeqnarray}{RCL}\label{F-4}
    &\tau_{{\rm THz},k} &= \int_{0}^{\infty }\int_{0}^{\infty} \frac{f_{{\rm{THz}}}^{(k)}\left(\rho \big| \frac{R}{K_{\rm{THz}}} \right)}{\Gamma\left(\mu\right)}  \notag \\
    \!\!\!&\!\!\!\times\!\!&\!\!\!\!\! {\Gamma\! \left(\!\mu, \mu  \left(  \frac{(2^{\frac{t}{B_{\rm THz}}}\!-\!1) \rho^2 \!\exp\left(\beta_{\rm{THz}} \rho \right) \sigma_{\rm{THz}}^2 }  {\frac{S}{K_{\rm{THz}}} G_{{\rm{THz}}} \eta_{{\rm{THz}}}} \right)^{\!\frac{\alpha}{2}} \!\right)}\! {\rm d} \rho {\rm d} t. \IEEEeqnarraynumspace 
\end{IEEEeqnarray}
Finally, by substituting equation (\ref{F-4}) into equation (\ref{F-1}), the proof of Theorem~\ref{theorem2} is completed.

\vspace{-2mm}\section{Proof of Theorem~\ref{theorem3}}\label{app:theorem3}
The initial steps of the proof for RF networks follow the same logic as for THz networks, and thus, they are omitted here. The conditional probability for RF networks, which corresponds to (\ref{F-3}) in the case of THz networks, can be calculated by
\begin{IEEEeqnarray}{RCL}    
    &\mathbb{P}& \left [ {\rm{SNR}}_{{\rm{RF}},k} > 2^{\frac{t}{{B_{\rm{RF}}}}}-1  \Big| r = \rho \right ] \notag \\
    &=&
    \mathbb{P} \left [ \mathcal{X}_{\rm{RF}} >  {\left(2^{\frac{t}{{B_{\rm{RF}}}}}-1\right) \frac{K_{\rm{RF}}}{S} \rho^{\beta_{\rm{RF}}}  \sigma_{\rm{RF}}^2 }  { G_{{\rm{THz}}}^{-1} \eta_{{\rm{THz}}}^{-1}}  \right ] \notag \IEEEeqnarraynumspace  \\
    & \overset{(a)}{=}& \exp\left( -{\left(2^{\frac{t}{{B_{\rm{RF}}}}}-1\right) \frac{K_{\rm{RF}}}{S} \rho^{\beta_{\rm{RF}}}  \sigma_{\rm{RF}}^2 } { G_{{\rm{THz}}}^{-1} \eta_{{\rm{THz}}}^{-1}} \right),
\end{IEEEeqnarray}
where (a) follows the definition of the \ac{CCDF} of $\mathcal{X}_{\rm{RF}}$. Therefore, the single-hop throughput $\tau_{{\rm RF},k}$, where $k=1,2$, is given by,
\begin{IEEEeqnarray}{RCL}   
& \tau_{{\rm RF},k} & = \int_{0}^{\infty }\int_{0}^{\infty} f_{{\rm{RF}}}^{(k)}\left(\rho \bigg| \frac{R}{K_{\rm{RF}}} \right)  \notag \\  &\times \exp & \left( -{(2^{\frac{t}{B_{\rm RF}}}\!-\!1) \frac{K_{\rm{RF}}}{S} \rho^{\beta_{\rm{RF}}}  \sigma_{\rm{RF}}^2 }  { G_{{\rm{RF}}}^{-1} \eta_{{\rm{RF}}}^{-1}} \right) {\rm d} \rho {\rm d} t. \IEEEeqnarraynumspace  
\end{IEEEeqnarray}

\vspace{-2mm}\section{Proof of Theorem~\ref{theorem4}}\label{app:theorem4}
As the hops are assumed to be independent, the coverage probability in multi-hop routing can be calculated using the following 
\begin{small}
\begin{IEEEeqnarray}{RCL} 
    &P_{\rm{THz}}^C & = \prod_{k=1}^{K_{\rm{THz}}} \mathbb{E}_\rho \left[ \mathbb{P}\left [ {\rm SNR}_{{\rm{THz}},k}> \gamma_{\rm{THz}} | r=\rho \right ] \right]  \notag \\
    &=& \!\!\!\! \left ( \int_{0}^{\infty} \!\! \!\! f_{{\rm{THz}}}^{(1)} \! \left(\! \rho \bigg| \frac{R}{K_{\rm{THz}}} \right)  \! \mathbb{P}\! \left [ {\rm SNR}_{{\rm{THz}}} \left(\!\rho, \frac{S}{{K_{\rm{THz}}}} \! \right) \!\!>\!\! \gamma_{\rm{THz}} \!\right ]\! {\rm d} \rho \right )^2  \notag \\
    &&\!\!\!\!\!\!\!\!\! \times \!\! \left ( \int_{0}^{\infty} \!\!\!\! f_{{\rm{THz}}}^{(2)}\!\left(\!\rho \bigg| \frac{R}{K_{\rm{THz}}} \right)  \! \mathbb{P}\! \left[ \!  {\rm SNR}_{{\rm{THz}}} \! \left( \! \rho, \frac{S}{{K_{\rm{THz}}}} \right) \!\!>\!\! \gamma_{\rm{THz}} \!\right ] \! {\rm d} \rho \!\right )^{\!\!K_{\rm{THz}}\!-\!2}  \notag \\
    & \overset{(a)}{=}& \!\!\!\! \left ( \int_{0}^{\infty}   \! \frac{f_{{\rm{THz}}}^{(1)}\left(\rho \big| \frac{R}{K_{\rm{THz}}} \right)}{\Gamma\left(\mu\right)} {\Gamma \left(\mu, \mu  \left( \mathcal{C}(\rho) \right)^{\frac{\alpha}{2}} \right)} {\rm d} \rho\right )^2  \notag \\
    &&\!\!\!\!\!\!\!\!\! \times \!\! \left (\int_{0}^{\infty}  \! \frac{f_{{\rm{THz}}}^{(2)}\left(\rho \big| \frac{R}{K_{\rm{THz}}} \right)}{\Gamma\left(\mu\right)} {\Gamma \left(\mu, \mu \left(  \mathcal{C}(\rho)  \right)^{\frac{\alpha}{2}} \right)} {\rm d} \rho\right )^{K_{\rm{THz}}-2}, \IEEEeqnarraynumspace
\end{IEEEeqnarray}
\end{small}where $\mathcal{C}(\rho) = \frac{K_{\rm{THz}} \gamma_{\rm{THz}} \rho^2 \exp\left(\beta_{\rm{THz}} \rho \right) \sigma_{\rm{THz}}^2 }  {S G_{{\rm{THz}}} \eta_{{\rm{THz}}}} $ is a function based on $ \rho $, and the proof of step (a) is similar to that of (\ref{F-3}), therefore, it is omitted here.

\bibliographystyle{IEEEtran}
\bibliography{references}

\end{document}